\ifdefined\pdfoutput\pdfoutput=1\fi
\documentclass[10pt, letterpaper]{IEEEtran}
\usepackage{setspace} 
\usepackage{cite}
\usepackage{graphicx}
\usepackage{caption}
\usepackage{subcaption} 
\usepackage[cmex10]{amsmath}
\usepackage{mathtools}
\usepackage{amsthm}
\usepackage{amsfonts}
\usepackage[font=small,skip=0pt]{caption}
\usepackage{amssymb}
\usepackage{textcomp}
\usepackage{gensymb}
\usepackage{xcolor}
\def\BibTeX{{\rm B\kern-.05em{\sc i\kern-.025em b}\kern-.08em
    T\kern-.1667em\lower.7ex\hbox{E}\kern-.125emX}}
\usepackage[linesnumbered,ruled,vlined]{algorithm2e}
\usepackage{bm,xstring}
\usepackage{tikz}
\usepackage{hyperref}
\usetikzlibrary{matrix}
\usetikzlibrary{shapes,arrows,positioning}
\usetikzlibrary{arrows.meta, positioning, quotes}
\usetikzlibrary{decorations.pathreplacing}

\DeclareMathOperator*{\argmin}{argmin}
\usepackage{color}

\newcommand{\ignore}[1]{}

\newtheorem{lemma}{Lemma}

\newtheorem{theorem}{Theorem}
\newtheorem{corollary}{Corollary}
\theoremstyle{definition}

\newtheorem{definition}{Definition}

\newtheorem{remark}{Remark}

\begin{document}

\title{Online Scheduling for Throughput Maximization of
Time-varying Markovian Channels with Unknown Statistics}

\author{Tasmeen~Zaman~Ornee,~\IEEEmembership{Member,~IEEE,}
        Clement~Kam,~\IEEEmembership{ Member,~IEEE,}
        and~Ness~B.~Shroff,~\IEEEmembership{Fellow,~IEEE}
\thanks{A part of this manuscript has been accepted by {\it IEEE INFOCOM}, 2026 \cite{orneeinfocom2026}.}
\thanks{{T. Z. Ornee and N. B. Shroff are with the Department of Electrical and Computer Engineering, The Ohio State University, Columbus, OH, 43210 USA (email: ornee.1@osu.edu, shroff.11@osu.edu)}}
\thanks{C. Kam is with the U.S. Naval Research Laboratory, Washington, DC, 20375 USA (e-mail: clement.kam@nrl.navy.mil).}
\thanks{This work has been supported in part by the Army Research Laboratory and was accomplished under Cooperative Agreement Number W911NF-23-2-0225, by the U.S. National Science Foundation under the grants: NSF AI Institute (AI-EDGE) 2112471, CNS-2312836, CNS-2225561, and CNS-2239677, and Office of Naval Research under grant N00014-24-1-2729. The views and conclusions contained in this document are those of the authors and should not be interpreted as representing the official policies, either expressed or implied, of the Army Research Laboratory or the U.S. Government. The U.S. Government is authorized to reproduce and distribute reprints for Government purposes notwithstanding any copyright notation herein.}
}

\maketitle

\begin{abstract} 

We consider a wireless scheduling problem in downlink wireless networks with unknown channel statistics, where a Base Station (BS) sends data to multiple users. The scheduling performance relies heavily on accurate Channel State Information (CSI), which is often costly to acquire. In this paper, CSI is obtained from ACK/NACK feedback, only after each scheduled transmission. Due to limited wireless channel resources, all users cannot be scheduled for transmission simultaneously. Hence, the most recently observed CSI can be outdated.  
The traditional approach to solve scheduling problems using outdated CSI is to utilize belief states, which are calculated using the time correlation statistics of channels. However, channel statistics are often unknown; consequently, belief states can be uncountable and this approach becomes infeasible. In this paper, we introduce a new sufficient statistics for the wireless scheduling problem. Towards this effort, we characterize the CSI staleness by the Age of Channel State Information (AoCSI) and show that the latest observed CSI and its AoCSI is a sufficient statistic of the history to make the scheduling decisions. Accordingly, we are able to reduce the state space for online learning. Our goal is to develop an online scheduling algorithm that maximizes the expected sum throughput of all users over a finite time-horizon while satisfying a channel resource constraint. The formulated problem is a Restless Multi-armed Bandit (RMAB). We develop an online Maximum Gain First (Online-MGF) policy, which achieves sub-linear regret on the number of episodes. For a special case of ON/OFF channels, we are able to prove indexability and derive a closed-form expression of the Whittle index. Numerical results demonstrate that the Online-MGF policy converges to MGF and Whittle index policies with known statistics within a very few episodes. In addition, Online-MGF outperforms Maximum AoCSI First (MAF) and random policies.

\end{abstract}


\section{Introduction}

6G and Future-Generation (Future-G) wireless networks aims to support a massive number of users and a broad range of applications, such as interconnected Internet of Things (IoT) devices, vehicle to vehicle and vehicle to infrastructure communications, and Unmanned Aerial Vehicle (UAV) deployments.  
A fundamental challenge in these applications is the efficient allocation of limited wireless resources such as spectrum, power, and time slots among a massive number of users. This difficulty is further increased by the highly dynamic nature of the wireless environment; high-mobility nodes (e.g., UAVs and vehicles) and the use of higher-frequency bands (mmWave and Terahertz) cause wireless channel states to fluctuate rapidly due to blockage, Doppler shifts, and fast fading, which yields the static resource allocation techniques ineffective. 

One strategy to address the difficulty that arises from fluctuating channel conditions is to exploit the time correlation of the channel state information (CSI) to predict the current CSI based on outdated CSI \cite{Ouyang_Infocom, Ouyang_TMC, Javidi_TIT, Neely_2010, zhao2008myopic, liu2002opportunistic, Liu_TIT}. For example, a BS can anticipate the future capacity of a channel by understanding its time-varying characteristics. 
However, accurate time correlation statistics for a wireless channel are often unknown in practice. The lack of complete channel statistical knowledge makes the resource allocation problem highly complicated. 

Motivated by these practical limitations, we investigate an online wireless scheduling problem for throughput maximization under unknown channel statistics and imperfect CSI in this paper.
Our system consists of a BS that transmits data to $N$ users and operates in a discrete-time environment. Due to limited communication resources, the BS selects $M$ out of $N$ users for transmission at any time-slot $t$. To model the time correlations in fading channels, we assume that the wireless channel between the BS and each scheduled user evolves as a discrete-time, finite-state Markov chain. 
The CSI can be modeled by Signal-to-Noise Ratio (SNR), Bit Error Rate (BER), received signal power, etc. 
The BS obtains the CSI from an ACK/NACK feedback after each scheduled transmission. Therefore, the BS utilizes the past available observations to predict the CSI in the current time-slot before making the scheduling decisions. Moreover, the channel statistics are unknown to the BS in our model. 

Therefore, in order to maximize the throughput, the BS requires to estimate the channel statistics and the current CSI based on the available information prior to making the scheduling decisions. 
To that end, we answer the following research question in this study:

\begin{quote}
{\it How to design an efficient online scheduling algorithm for maximizing throughput while satisfying an instantaneous resource constraint with unknown channel statistics and imperfect CSI?
} 
\end{quote}

The contributions of this study are summarized below:

\begin{itemize}

\item The optimal scheduling problem for maximizing throughput can be formulated as an RMAB \cite{whittle_restless}, where each user (and its channel) is represented as an arm and each arm is described as an MDP \cite{whittle_restless}. 
Existing studies \cite{Liu_TIT, Ouyang_TMC, Ouyang_Infocom, Javidi_TIT, Ansell_2003, Neely_2010, chen2021scheduling, chen2022index} utilize belief states as sufficient statistics  which can lead to an uncountable number of possible states, specifically when the channel statistics are unknown. We are able to reduce the state space of this problem by utilizing a sufficient statistic of the history \cite{bertsekas2011dynamic}. To that end, we leverage a metric called \emph{Age of Channel State Information (AoCSI)} \cite{costa2015age}, which is defined as the time difference between the current time and the last time CSI is updated. We show that the latest observed CSI and its AoCSI form a sufficient statistic of the information history to make the scheduling decisions for each user (see Theorem \ref{theorem_simplification}). Therefore, by leveraging the latest observed CSI and its AoCSI as states, we replace the Partially Observed Markov Decision Process (POMDP) or belief MDP framework used in \cite{Liu_TIT, Ouyang_TMC, Ouyang_Infocom, Javidi_TIT, Ansell_2003, Neely_2010, chen2021scheduling, chen2022index} with a more tractable MDP framework to solve the scheduling problem with unknown channel statistics.

\item By using (i) relaxation and  Lagrangian decomposition and (ii) an optimistic Upper Confidence Bound (UCB) approach, we develop an online Maximum Gain First (Online-MGF) policy (see Algorithm \ref{algo1}). Unlike the Whittle index-based policy \cite{whittle_restless}, our policy does not require satisfying any indexability condition. Utilizing standard relaxation and Lagrangian decomposition, we decouple the original problem into multiple independent MDPs. Because we are able to reduce the state space of each MDP by finding a sufficient statistic, we can simplify the search space of the online learning algorithm. Unlike existing online learning approaches for RMABs \cite{wang2024online, wang2023optimistic}, our Online-MGF policy leverages AoCSI and do not need to estimate the entire transition probabilities (see Remarks \ref{remark1} and \ref{remark2}). It eliminates the $\mathcal{O}(N\tau^2|\mathcal{C}|^2)$ complexity of estimating full transition matrices for passive actions \cite{wang2024online} and the $\mathcal{O}(N\tau^3|\mathcal{C}|^3)$ complexity of root-finding step required by current online Whittle index policies \cite{wang2023optimistic}. This results in a more efficient per-episode complexity of $\mathcal{O}(NH\tau|\mathcal{C}|^2 + HN \log N)$ (see Remark \ref{remark3}).


\item We propose a stronger definition for regret that evaluates the performance of the proposed policy with respect to the optimal policy. Earlier studies provide regret bound for RMAB with respect to the Lagrangian function \cite{orneeinfocom2026, wang2023optimistic, shisher2023learning}. Our definition is more general and the proposed Online-MGF policy achieves sublinear regret $O(\sqrt{K \log K})$ on the number of episodes $K$ (see Theorems \ref{theorem_regret}, \ref{theorem_regret_2}, and \ref{theorem_regret_3}).

\item For the special case of a two-state Markov chain, such as an ON/OFF channel, the Online-MGF policy in Algorithm \ref{algo1} can be simplified. This is because the feedback allows us a direct mapping to the actual channel state (e.g., a NACK confirms the OFF state with probability 1). Consequently, we only need to estimate transition probabilities associated with the ON states (see Algorithm \ref{algo2}). Furthermore, we establish indexability for this setting and derive the closed-form expression of the Whittle index. The algorithm for the Online Whittle index policy is provided in Algorithm \ref{algo_3}.

\item Numerical results show that the Online-MGF policy converges faster to the MGF and Whittle index policies with known channel statistics and achieves good performance over Maximum AoCSI First (MAF) and Random policies (see Figures \ref{fig:simulationH}-\ref{fig:simulationM}). Moreover, Online-MGF policy shows similar performance as the Whittle index policy with known channel statistics (see Figures \ref{fig:simulationH}-\ref{fig:simulationM}). 
In one setup, the Online-MGF policy outperforms the Whittle index policy with known channel statistics (see Figure \ref{fig:simulationP}). 
 

\end{itemize}


\section{Related Work}

Opportunistic scheduling algorithms have been widely used and extensively studied in the literature under perfect and imperfect CSI \cite{liu2002opportunistic, tassiulas2002scheduling,  zhao2008myopic, Liu_TIT, Ouyang_Infocom, Javidi_TIT, Neely_2010, Ansell_2003}. These studies addressed the scheduling problem using the POMDP (or belief MDP) formulation which considers belief states \cite{liu2002opportunistic, tassiulas2002scheduling, zhao2008myopic, Liu_TIT, Ouyang_TMC, Javidi_TIT,  Neely_2010, Ansell_2003}. Such formulations can lead to an uncountable number of state spaces, specifically when the channel statistics are unknown. Age of Information (AoI) has emerged as an important metric to study sampling and scheduling problems, because it captures the ``freshness (or rather staleness)" of information  \cite{sun2019sampling, Chen_2022,  Shisher_2024, ornee2023whittle, xiong2022index, kadota2018scheduling, chen2021scheduling, Ornee2021, Wiener_TIT, ornee2025safety}. Building upon this concept, \emph{Age of Channel State Information (AoCSI)} was proposed in \cite{costa2015age} to analyze the impact of CSI staleness on the overall system performance. Subsequently, several papers \cite{costa2015age_2, truong2013effects, lipski2024age, chakraborty2025send} studied different scheduling schemes and the impact of AoCSI and different functions of AoCSI. When the channel statistics are unknown, using POMDP (or belief MDP) formulations with belief states makes the solution computationally intractable and introduces significant challenges for online learning algorithms. Unlike \cite{liu2002opportunistic, tassiulas2002scheduling,  zhao2008myopic, Liu_TIT, Ouyang_Infocom, Javidi_TIT,  Neely_2010, Ansell_2003}, we do not require utilizing POMDP (or belief MDP) formulation with belief states. We obtain a sufficient statistic of the history in this study, which contains the last observed CSI and its AoCSI. Therefore, instead of utilizing belief states, we can use a  smaller state space.  

\begin{figure}
\centering
\includegraphics[width=8cm]{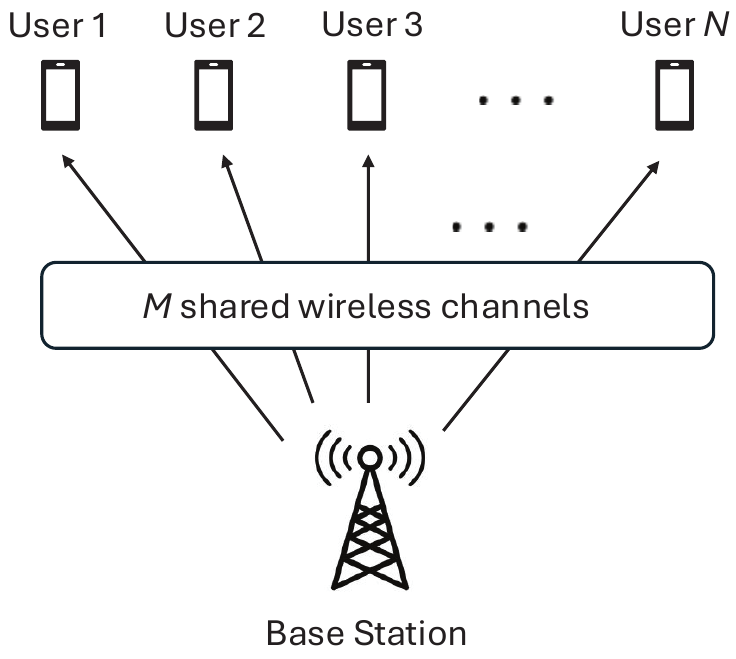}
\vspace*{1.0mm} 
\caption{System Model.}\vspace{-0.0cm}
\label{model}
\end{figure}
 
Moreover, restless bandits are a popular framework to solve resource allocation and sequential decision-making problems. Numerous scheduling problems have been formulated as RMABs and Whittle index policies were found for contexts such as multi-access channels, network utility maximization, age penalty minimization, remote estimation, etc \cite{whittle_restless, Ouyang_TMC,  Chen_2022, Shisher_2024, ornee2023whittle}. Whittle index policy is proven to be asymptotically optimal \cite{weber1990index} under an indexability condition, which is often very challenging to establish in practice. Consequently, non-indexable scheduling policies have been analyzed in recent years \cite{xiong2022index, zou2021minimizing, chen2022index, shisher2023learning, ornee2023context, ornee2025safety}. Unlike \cite{xiong2022index, zou2021minimizing, chen2022index, shisher2023learning, ornee2023context, ornee2025safety}, we consider unknown channel statistics, and thus the methods established in existing studies cannot be used to solve our problem. A gain index-based asymptotically optimal policy was proposed in  \cite{chen2022index, shisher2023learning, ornee2023context, ornee2025safety, shisher2025computation} with known system dynamics. 
When system dynamics are unknown, online learning is a promising approach to solve RMABs, where a decision-maker simultaneously estimates the unknown system dynamics and makes scheduling decisions.
There exist several online learning policies that studied RMAB problems with unknown system dynamics \cite{neu2013efficient, foster2020beyond, li2019combinatorial, ortner2012regret, akbarzadeh2023learning,wang2023optimistic, shisher2025onlinelearningwhittleindices}. Without the ``restless" setting, \cite{neu2013efficient, foster2020beyond, li2019combinatorial} studied online learning policies for Multi-armed Bandits (MABs). Under indexability, \cite{akbarzadeh2023learning} proposed a Thompson sampling-based online Whittle index policy and \cite{wang2023optimistic, shisher2025onlinelearningwhittleindices} proposed UCB-based online Whittle index policies. Unlike \cite{akbarzadeh2023learning, wang2023optimistic, shisher2025onlinelearningwhittleindices}, we develop a UCB-based online Maximum Gain First (Online-MGF) policy that does not require satisfying any indexability condition, which often does not hold in practice. Moreover, compared to the existing studies, our proposed Online-MGF policy requires us to estimate probabilities associated with one action instead of the entire transition probabilities as we find a sufficient statistic of the history.   

\section{Model and Problem Formulation}
\subsection{System Model}
Our system consists of a wireless network with one Base Station (BS) and $N$ users with $M$ shared wireless channels as depicted in Figure \ref{model}. The system is time-slotted and operates within a finite time horizon $T$. At any time-slot $t$, the BS can simultaneously transmit to at most $M$ users without interference, where $M <N$. 

The wireless channel $C_n (t)$ between the BS and user $n$ is independent across users, i.e., $C_n (t), t=1, 2, \ldots, T$, and $C_m (t), t=1, 2, \ldots, T$ are independent for all $n \neq m$. The CSI $C_n (t)$ remains static within each time-slot. However, $C_n (t)$ evolves stochastically between successive time-slots. 
The CSI $C_n (t)$ is modeled as a discrete-time, finite-state time-homogeneous Markov chain with $C$ states, 
where 
$2 \leq C < \infty$. 
For each channel $n$, $C_n (t)$ evolves according to a transition matrix $P_n \in \mathbb{R}^{N \times N}$, where
\begin{align}
P_n (C_n (t) =j |C_n (t-1) =i) = p_{ij},
\end{align}
which is the $(i, j)$-th entry of the transition matrix ${P}_n$. The transition matrix ${P}_n$ is unknown in our model.

At the beginning of each time-slot $t$, the BS does not know the actual CSI $C_n (t)$ while making the scheduling decisions. Let $a_n (t) \in \{0, 1\}$ denotes the scheduling decision to schedule user $n$ at time-slot $t$, defined as
\begin{align} \label{action}
a_n (t) \!\!=\!
\begin{cases}
& \!\!\!\!\!\!\!1, \!~\text{if user $n$  is scheduled in time-slot $t$},\\
& \!\!\!\!\!\!\!0, \!~\text{otherwise.}
\end{cases}
\end{align}
Whenever user $n$ receives data from the BS, it sends a zero-delay feedback to the BS to inform a successful or failed transmission. We denote $\beta_n(t) \in \{0, 1\}$ as the delivery indicator for user $n$ at time-slot $t$. If the transmission fails (i.e., $\beta_n(t) = 0$), the BS recognizes that the channel is in a bad condition or "OFF" state. However, if the transmission is successful (i.e., $\beta_n(t) = 1$), the user also sends the received signal power, which allows the BS to calculate the Signal-to-Noise Ratio (SNR). Because the CSI is modeled as a finite-state Markov chain, this feedback information is essential for the BS to map the successful transmission to a specific channel state $C_n(t)$ within the discretized CSI space. However, the feedback is only received at the end of the time-slot after the data is transmitted.
The delivery indicators $\beta_n (t)$ are independent across both users and time-slots. Consequently, when the BS makes the scheduling decision in time-slot $t$, the most recent information available of user $n$ is the last observed CSI $C_n (t-\Delta_n (t))$, which is observed at time-slot $t-\Delta_n (t)$. The time difference between the current time $t$ and the last CSI update time $t-\Delta_n (t)$ is called the \emph{Age of Channel State Information (AoCSI)} \cite{costa2015age},  defined as
\begin{align} \label{age}
\Delta_n (t+1) =
\begin{cases}
& \!\!\!\!\!\! 1, ~~~~~~~~~~~\text{if} {\thinspace} a_n (t) =1{\thinspace}, \\
& \!\!\!\!\!\! \Delta_n (t) +1, \text{otherwise.}
\end{cases}
\end{align}


\subsection{Problem Formulation with Known Channel Statistics}

Let $\pi = (a_n (1), a_n (2), \ldots, a_n (T))_{n=1}^{N}$ denote a scheduling policy, where $a_n (t)$ is defined in \eqref{action} for all $t=1, 2, \ldots, T$. Our goal is to find the optimal policy $\pi$ that maximizes the sum of throughput over the finite time horizon $T$, subject to the constraint on the number of users selected in each time-slot $t$. 
The sum-throughput optimization problem is formulated as follows
\begin{align}
& \max_{\pi \in \Pi} \sum_{t=1}^{T} \sum_{n=1}^{N} \mathbb{E}_{\pi} \bigg[   h( C_n (t), a_n (t)) \bigg| \mathcal{H}_n (t)\bigg] \label{problem} \\
&~ \text{s.t.} ~\sum_{n=1}^{N} a_n (t) \leq M, ~t=1, 2, \ldots, T, \label{constraint}
\end{align}
where $\Pi$ is the set of all causal scheduling policies, $h (C_n (t), a_n (t))$ is the instantaneous throughput of 
user-$n$ 
in time-slot $t$, and $\mathcal{H}_n (t)$ is the information history available for the BS at time-slot $t$, defined as
\begin{align}
\mathcal{H}_n (t) \!=\!  (C_n (\tau\!-\!\Delta_n (\tau)), \Delta_n (\tau), a_n (\tau-1), \beta_n (\tau -1))_{\tau=1}^{t},
\end{align}
which includes all previously observed channel states, their AoCSI values, scheduling decisions, and delivery indicators of user $n$ up to time-slot $t$. The instantaneous throughput $h (C_n (t), a_n (t)) =0$, if user $n$ is not scheduled at time-slot $t$ (i.e., $a_n (t) =0$). When user $n$ is scheduled at time-slot $t$ (i.e., $a_n(t) = 1$), the throughput depends on the CSI $C_n(t)$. In this framework, $C_n(t)$ can represent discretized Signal-to-Noise Ratio (SNR) levels or received signal power. For instance, if $C_n(t)$ represents discretized SNR levels, the instantaneous throughput is defined by the Shannon capacity
\begin{align}
h(C_n(t), a_n(t)) = a_n(t) W \log_2(1 + C_n(t)),
\end{align}
where $W$ represents the channel bandwidth. Alternatively, if $C_n(t)$ is the received signal power, the throughput is given by
\begin{align}
h(C_n(t), a_n(t)) = a_n(t) W \log_2(1 + \frac{C_n(t)}{\sigma^2}),
\end{align}
where $\sigma^2$ is the noise power.

Problem \eqref{problem}-\eqref{constraint} can be formulated as a Partially Observed Markov Decision Process (POMDP) \cite{krishnamurthy2016partially}. This is because the scheduler has no knowledge of the current CSI at time-slot $t$ while making the scheduling decisions. Existing studies utilize belief states to solve similar kind of POMDP problems \cite{Ouyang_TMC, Ouyang_Infocom, Javidi_TIT, Ansell_2003, Liu_TIT, Neely_2010, chen2021scheduling, chen2022index}. However, this approach requires handling over an uncountably infinite state space of continuous probability distributions. Furthermore, when the underlying channel statistics are unknown, evaluating these belief states become significantly challenging and computationally intractable for online learning algorithms.
To address this, we simplify  problem \eqref{problem}-\eqref{constraint} by reducing its state space to a countable space. We achieve this by using a sufficient statistic of the history to provide an equivalent MDP formulation. The details are presented in Section \ref{simplification}.


\section{Problem Simplification} \label{simplification}

We are able to simplify problem \eqref{problem}-\eqref{constraint} by utilizing the sufficient statistic of the history \cite{bertsekas2011dynamic}. 


\begin{theorem} \label{theorem_simplification}
If $C_n (t) \leftrightarrow C_n (t-1) \leftrightarrow C_n (t-2) \leftrightarrow \ldots$ is a Markov chain, then the last observed CSI and its AoCSI $(C_n (t-\Delta_n (t)), \Delta_n (t))_{n=1}^N$ is a sufficient statistic of the information history $(\mathcal{H}_n (t))_{n=1}^{N}$ for making the scheduling decisions $(a_n (t))_{n=1}^N$ at time-slot $t$ in problem \eqref{problem}-\eqref{constraint}.
\end{theorem}

\begin{proof}
See Appendix \ref{proof_theorem_1}.
\end{proof}

Theorem \ref{theorem_simplification} implies that $(C_n (t-\Delta_n (t)), \Delta_n (t))_{n=1}^N$ is a sufficient statistic to find the optimal scheduling decisions at time-slot $t$. Using Theorem \ref{theorem_simplification}, we can reformulate problem \eqref{problem}-\eqref{constraint} as
\begin{align}
& \max_{\pi \in \Pi} \sum_{t=1}^{T} \sum_{n=1}^{N} \mathbb{E}_{\pi} \!\bigg[ \!f (C_n (t\!-\!\Delta_n (t)), \Delta_n (t), a_n (t))\!\bigg] \label{problem_simple} \\
&~ \text{s.t.} ~\sum_{n=1}^{N} a_n (t) \leq M,~t=1,2,\ldots, T, \label{constraint_simple}
\end{align}
where 
\begin{align}
& f (C_n (t-\Delta_n (t)), \Delta_n (t), a_n (t)) \nonumber\\
=& \mathbb{E} [h (C_n (t), a_n (t))| C_n (t-\Delta_n (t)), \Delta_n (t)].
\end{align}
If $a_n (t)=0$, then $f (C_n (t-\Delta_n (t)), \Delta_n (t), 0)=0$ as $h (C_n (t), 0) = 0$; If $a_n (t)=1$, then $f (C_n (t-\Delta_n (t)), \Delta_n (t), 1)$ can be further simplified: We denote the transition probability to the current CSI $C_n(t)=c'$ from the outdated CSI $C_{n}(t-\Delta_n(t))=c$ and AoCSI $\Delta_n(t)=\delta$ by $P_n(c'|c, \delta)$. Given $C_n (t-\Delta_n (t))=c$ and $\Delta_n(t)=\delta$, we have 
\begin{align} \label{func_f_a}
f (c, \delta, 1) = \sum_{c'} h (c', 1) P_n(c'|c, \delta),
\end{align}
where $P_n(c'|c, \delta)$ is obtained from the $\delta$-step transition matrix $P_n ^{\delta}$ and
\begin{align}
P_n (C_n (t) = j | C_n (t-\delta) = i) = p_{ij} ^{\delta}.
\end{align}

In problem \eqref{problem_simple}-\eqref{constraint_simple}, the state space is reduced from an uncountable number of belief states to a countable number of states. Moreover, the state space can be made finitely countable using a large truncated AoCSI value \cite{Ouyang_Infocom, Shisher_2024}. We let $\Gamma$ denote the maximum AoCSI value in the truncated space.

Problem \eqref{problem_simple}-\eqref{constraint_simple} is an RMAB where each user $n$ is represented as an arm and each arm $n$ is described as an MDP. Each MDP associated with arm $n$ has two actions at every time-slot $t$: active ($a_n (t) =1$) and passive ($a_n (t) =0$). The problem \eqref{problem_simple}-\eqref{constraint_simple} is ``restless" because the AoCSI $\Delta_{n}(t)$ evolves even when the $n$-th
arm is passive \cite{Shisher_2024}. Since we are able to simplify the original problem \eqref{problem}-\eqref{constraint}, the RMAB \eqref{problem_simple}-\eqref{constraint_simple} has a reduced state space compared to \eqref{problem}-\eqref{constraint}. 

However, even with reduced state space, solving RMAB problems is significantly challenging, and it is intractable to find the optimal solution \cite{PSPACE_book}. A Whittle index policy is known to be an efficient approach to solve RMAB problems \cite{whittle_restless}. This policy is known to be asymptotically optimal under a complicated condition called indexability \cite{weber1990index}. Nevertheless, indexability is often very difficult to establish in practice. Recently, another policy called the gain index policy has been proven to be asymptotically optimal for RMABs \cite{verloop2016asymptotically, chen2022index, ornee2025safety}. A key advantage of the gain index-based policy is that it does not require proving indexability. In this paper, we provide a gain index-based policy for the problem \eqref{problem_simple}-\eqref{constraint_simple}. Furthermore, in Section \ref{sim}, we show numerically that our developed gain index-based  policy performs equally well compared to the Whittle index policy.



In addition, the transition matrix $P_n$ is unknown in our model. Therefore, we first study a gain index-based policy for problem \eqref{problem_simple}-\eqref{constraint_simple} with known channel state transition matrix $\mathbb{P}_n$. Utilizing the insights from the solution with known channel statistics, we will develop an online algorithm called ``Online Maximum Gain First (Online-MGF)" policy when the channel statistics are unknown. To that end, we utilize Lagrangian dual decomposition and decompose problem \eqref{problem_simple}-\eqref{constraint_simple} to $N$ independent per-arm problems, detailed in Section \ref{decoupled_section}.



\subsection{Relaxation and Lagrangian Decomposition} \label{decoupled_section}

Following the standard relaxation and Lagrangian decomposition technique for RMABs \cite{whittle_restless}, we first relax the constraint \eqref{constraint_simple}
and obtain the following relaxed problem:
\begin{align}
& \max_{\pi \in \Pi} \sum_{t=1}^{T} \sum_{n=1}^{N} \mathbb{E}_{\pi} \!\bigg[ \!f (C_n (t-\Delta_n (t)), \Delta_n (t), a_n (t))\!\bigg] \label{relax_problem} \\
&~ \text{s.t.}~\sum_{t=1}^{T} \sum_{n=1}^{N} \mathbb{E}_{\pi} [ a_n (t)] \leq M T, \label{constraint_relax}   
\end{align}
where the relaxed constraint \eqref{constraint_relax} needs to be satisfied in total finite time $T$, instead of satisfying at every time-slot $t$. Next, we apply Lagrange multiplier $\lambda \geq 0$ (also known as dual cost) to the relaxed constraint \eqref{constraint_relax} and obtain the following Lagrangian dual problem:
\begin{align} \label{lagrangian}
 L ({\pi^ *}, {\mathbf{P}}, \lambda)
 \!=\!&\max_{\pi \in \Pi} \!\sum_{n=1}^{N} \!\mathbb{E}_{\pi} \!\bigg[\! \sum_{t=1}^{T}  \!f (C_n (t\!-\!\Delta_n (t)), \Delta_n (t), a_n (t)) \nonumber\\
 &~~~~~~~~~~~~~~~~~~~- \lambda (a_n (t) - MT) \bigg],
\end{align}
where $\mathbf{P} = (P_n)_{n=1}^{N}$. The term $\lambda MT$ is a constant and does not affect the policy $\pi$, and therefore can be removed. For a given $\lambda$, problem \eqref{lagrangian} can be decomposed into $N$ independent per-arm problems, where each per-arm problem associated with arm $n$ is given by
\begin{align} \label{decoupled}
&L_n ({\pi_n^ *}, P_n, \lambda) \nonumber\\ 
=& \!\max_{\pi_n \in \Pi_n} \!\!\mathbb{E}_{\pi_n} \!\!\bigg[\! \sum_{t=1}^{T} f (C_n (t\!-\!\Delta_n (t)),\Delta_n (t)), a_n (t))- \lambda a_n (t)\bigg],
\end{align}
where $L_n ({\pi_n^ *}, P_n, \lambda)$ is the optimum value of \eqref{decoupled}, $\pi_n= (a_n (1), 
a_n (2), \ldots, a_n (T))$ is the sub-scheduling policy of arm $n$, and $\Pi_n$ is the set of all causal sub-scheduling policies of arm $n$. 

We first solve \eqref{decoupled} for given $\lambda$. Following the solution to the decoupled problem \eqref{decoupled}, we will develop a solution for the original RMAB \eqref{problem_simple}-\eqref{constraint_simple}. Then we obtain the optimal Lagrange multiplier from $\lambda^* = \argmin_{ \lambda} L(\pi^*, \bf{P}, \lambda)$, where $L(\pi^*, \bf{P}, \lambda)$ is defined in \eqref{lagrangian}.

Using dynamic programming \cite{bertsekas2011dynamic}, \eqref{decoupled} can be solved. Given $C_n (t-\Delta_n (t)) =c$ and $\Delta_n (t) = \delta$ at time-slot $t$, the Bellman optimality equation for the MDP \eqref{decoupled} is given by 
\begin{align}
V_{n, t} (c, \delta; \lambda) = \max_{a \in \mathcal{A}} Q_{n, t} (c, \delta, a; \lambda),
\end{align}
where $Q_{n, t} (c, \delta, a; \lambda)$ is the action-value function defined as
\begin{align}
& Q_{n, t} (c, \delta, a; \lambda) = \nonumber\\
& \begin{cases}
& \!\!\!\!\!\!\! f (c, \delta, 1) \!+\!\sum_{c' \in \mathcal{C}} P_n (c'|c, \delta) V_{n, t+1} (c', 1; \lambda) \!-\! \lambda,
\text{if} {\thinspace} a=1, \\
& \!\!\!\!\!\!\!  V_{n, t+1} (c, \delta+1; \lambda), ~~~~~~~~~~~~~~~~~~~~~~~~~~~~~~~~~~\text{if} {\thinspace} a=0. \label{ac_value}\\
\end{cases}
\end{align}
where $V_{n, T+1}(c, \delta)=0$ for all $c, \delta$.
Using \eqref{func_f_a}, the action-value function \eqref{ac_value} can be further simplified as
\begin{align} \label{value_simple}
& Q_{n, t} (c, \delta, a; \lambda) = \nonumber\\
& \begin{cases}
& \!\!\!\!\!\!\! \sum_{c'} P_n (c'|c, \delta) (h (c',1)+V_{n, t+1} (c', 1; \lambda)) - \lambda,
\text{if} {\thinspace} a=1, \\
& \!\!\!\!\!\!\!  V_{n, t+1} (c, \delta+1; \lambda), ~~~~~~~~~~~~~~~~~~~~~~~~~~~~~~~~~\text{if} {\thinspace} a=0.\\
\end{cases}
\end{align}
The value function $V_{n, t} (c, \delta; \lambda)$ can be calculated using the backward induction method \cite{bertsekas2011dynamic} from time $t=T, T-1, \ldots, 1$:
\begin{align}
    &V_{n, t}(c, \delta; \lambda)=\nonumber\\
    &\max\bigg\{\sum_{c'} P_n (c'|c, \delta) (h (c',1)+V_{n, t+1} (c', 1; \lambda)) - \lambda, \nonumber\\
    &~~~~~~~~~~~~V_{n, t+1} (c, \delta+1; \lambda)\bigg\},
\end{align}
where $V_{n, T+1} (c, \delta; \lambda)=0$ for all $(c, \delta)$.

If the action-value function associated with the active action $Q_{n, t} (c, \delta, 1; \lambda)$ is higher than the action-value function associated with the passive action $Q_{n, t} (c, \delta, 0; \lambda)$, i.e, $Q_{n, t} (c, \delta, 1; \lambda) > Q_{n, t} (c, \delta, 0; \lambda)$, then the optimal decision to problem \eqref{decoupled} is to schedule arm $n$. 

An equivalent policy is to employ the gain index, which is defined by \cite{chen2022index, shisher2023learning, ornee2025safety}
\begin{align} \label{gain}
\alpha_{n, t} (c, \delta) =  Q_{n, t} (c, \delta, 1; \lambda) - Q_{n, t} (c, \delta, 0; \lambda).   
\end{align} 
The gain index $\alpha_{n, t} (c, \delta)$ is the difference between two action-value functions. Following the policy with the gain index $\alpha_{n, t} (c, \delta)$, arm $n$ is scheduled if the gain index is positive, i.e., $\alpha_{n, t} (c, \delta) > 0$. 
To that end, we develop a Maximum Gain First (MGF) policy. This policy selects at most $M$ arms with the highest gain.

\begin{remark}\label{remark1}
{\it Analysis of the action-value function \eqref{value_simple} yields useful insights to design our online learning algorithm. From \eqref{value_simple}, if the action is passive (i.e., $a=0$), it is evident from \eqref{age} that the AoCSI transition will be from $\delta$ to $\delta+1$ and the last observed CSI $c$ will remain in the same state $c$. Consequently, the actual state transition is known to the BS with probability $1$ when a passive action is taken. Therefore, to design our online algorithm, we do not need to compute any probabilities associated with the passive action.}
\end{remark}

\section{Problem Statement for Online Setting}\label{onlinePolicy}


In our online learning environment, the BS interacts with $N$ arms through $K$ episodes,
where the length of the time horizon in each episode is $H$ time-slots. 
A key challenge in this setting is that the decision maker has no prior knowledge of the actual transition probabilities $\mathbf{P} = ({P}_n)_{n=1}^{N}$. In every episode $k$, the BS starts from an initial state $({\mathbf{c}} (1), {\boldsymbol{\delta}} (1)) = \{(c_1 (1), \delta_1 (1)), (c_2 (1), \delta_2 (1)), \ldots, (c_N (1), \delta_N (1))\}$ and implements a policy $\pi^ k$. Subsequently, it utilizes the observations up to $H$ time-slots to iteratively estimate the underlying transition probabilities. 

Finding the optimal policy for RMAB problems is generally intractable and PSPACE-hard \cite{PSPACE_book}. Therefore, it is significantly challenging  to evaluate the performance of an online learning policy with respect to the optimal policy of RMAB. The authors in \cite{orneeinfocom2026, wang2023optimistic} employ Lagrangian dual problem to evaluate the performance of the developed online policies. Recently, \cite{shisher2025onlinelearningwhittleindices} utilize the optimal policy to measure the regret for an online Whittle index policy. Similar to \cite{shisher2025onlinelearningwhittleindices}, we evaluate the performance of our online learning algorithm $\pi^ k$ in episode $k$ with respect to the optimal policy. 
Through simulation, we demonstrate that the Online-MGF policy converges to the MGF policy with known channel statistics within a very few episodes (see Section \ref{sim}). 

Let $J^k(\pi, \mathbf{P})$ denote the expected cumulative throughput obtained by executing policy $\pi$ during episode $k$ under the actual transition probabilities $\mathbf{P}$. Because each episode consists of $H$ time-slots, $J^k(\pi, \mathbf{P})$ is defined as the sum of the expected instantaneous throughput:
\begin{align}
J^k(\pi, \mathbf{P}) \!=\!
\mathbb{E}_{\pi, \mathbf{P}} \bigg[\!\sum_{t=1}^{H} \sum_{n=1}^N \!f(C_n(t-\Delta_n(t)), \Delta_n(t), a_n(t)) \bigg],
\end{align}
where the initial state at $t=1$ for episode $k$ is given by $(c(1), \delta(1))$. In this sequel, we define the regret.


\begin{definition} \label{def1}
\textbf{Regret.} Given the actual transition probabilities $\mathbf{P}$, the regret of the policy $\pi^{k}$ in episode $k$ 
is given by
\begin{align}
\text{Reg} (k) 
=& J^k (\pi^*, \mathbf{P}) - J^k (\pi^k, \mathbf{P}), \label{reg_new}
\end{align}
where $\pi^{*}$ is the optimal policy associated with transition probability $\mathbf{P}$.
Using \eqref{reg_new}, we get the total regret as follows:
\begin{align}
\text{Reg} (K) = \sum_{k=1}^{K} \text{Reg}(k). \label{total_regret_given}
\end{align}
\end{definition}
We get an upper bound of \eqref{reg_new} as follows
\begin{align}
\text{Reg} (k) 
=& J^k (\pi^*, \mathbf{P}) - J^k (\pi^k, \mathbf{P}) \\
\leq &  \sum_{n=1}^{N} L_n ({\pi_n ^{*}}, P_n, \lambda^*) -  J^k (\pi^k, \mathbf{P}) \label{reg_L}\\
=& \sum_{n=1}^{N} L_n({\pi_n ^{*}}, P_n, \lambda^*) - \sum_{n=1}^{N} L_n(\pi_n ^{k}, P_n, \lambda^*) \nonumber\\
& + \sum_{n=1}^{N} L_n(\pi_n ^{k}, P_n, \lambda^*) - J^k (\pi^k, \mathbf{P}), \label{reg_new_1}
\end{align}
where $\lambda^{*}$ minimizes the Lagrangian $L({\pi^{*}}, \mathbf{P}, \lambda)$, and \eqref{reg_L} holds because the Lagrangian upper bounds the original problem.

The regret bound obtained in 
\cite{orneeinfocom2026, wang2023optimistic} consists of the first two terms of \eqref{reg_new_1} that represents the performance difference associated with Lagrangian problem. Hence, Definition \ref{def1} is stronger compared \cite{orneeinfocom2026, wang2023optimistic} because the optimal policy is taken into account.


\section{Design of Online Policy for Solving \eqref{problem_simple}-\eqref{constraint_simple}} \label{policy}

We solve the online learning problem in two steps: (i) Utilizing a UCB-based approach, we first develop an Online-MGF policy in Section \ref{online_MGF} that does not need to satisfy any indexability
condition; (ii) Next, we analyze the regret and find a regret bound which illustrates that the developed algorithm achieves sublinear regret in Section \ref{regret_analysis}.

\subsection{Online-MGF Policy} \label{online_MGF}

We utilize UCB-based online learning to provide an algorithm for the Online-MGF policy. To compute the gain index in \eqref{gain}, we need to know the action-value functions, and, hence, the transition probabilities $\mathbf{P}$. Because $\mathbf{P}$ is unknown in our model, we develop an online learning method in Algorithm \ref{algo1}.

\subsubsection{Confidence Bound for Transition Probabilities}

For every arm $n$ and every episode $k$, we maintain the number of visits $B_n ^{k} (c, \delta, c')$, which denotes the number of transitions from CSI $C_{n}(t-\Delta_n(t))=c$ and AoCSI $\Delta_n(t)=\delta$ to the CSI $C_n(t)=c'$ under the active action $a_n(t)=1$ observed by the last $k$ episodes. Because our learning algorithm only needs to estimate probabilities associated with the active actions, we need to observe the state transitions under the active actions only.
For a small given constant $\mu > 0$, we define the confidence radius as follows:
\begin{align} \label{radius}
\rho_n ^{k} (c, \delta) = \sqrt{\frac{2 |\mathcal{C}| \log (2 \Gamma |\mathcal{C}| N \frac{k^2}{\mu})}{\max(1, B_n ^{k} (c, \delta))}},
\end{align}
where $B_n ^{k} (c, \delta) = \sum_{c' \in \{0, 1\}} B_n ^{k} (c, \delta, c')$ is the number of visits to state $(c, \delta)$ for arm $n$ under the active action by episode $k$, $\Gamma$ is the maximum value in the truncated AoCSI space, and $N$ is the number of users.

Let $P_n^k$ denote a candidate transition probability matrix for arm $n$ in episode $k$. The empirical transition probability is given by
\begin{align}
\hat P_n ^{k} (c'|c, \delta) = \frac{B_n ^{k} (c, \delta, c')}{B_n ^{k} (c, \delta)}.
\end{align}
Then, the confidence region of possible transition probabilities is given by
\begin{align} \label{ball_p}
\mathcal{B} ^{k} = 
& \bigg\{P_n ^k | \sum_{c' \in \mathcal{C}} \bigg|{P_n ^k (c' |c, \delta) - \hat P_n ^{k} (c' |c, \delta)}\bigg| \nonumber\\
& ~~~~~~~~~~~~~~~~~~~~~~~\leq \rho_n ^{k} (c, \delta), \forall n, c, \delta\bigg\},
\end{align}
where $P_n^k$ represents any candidate transition probability matrix for arm $n$ in episode $k$ that falls within the allowable confidence radius around the empirical estimate $\hat{P}_n^k$. In the following section, we formulate an optimization problem where our algorithm explicitly selects $P_n^k \in \mathcal{B}^k$ that maximizes the expected future reward. Thus this step transforms the candidate transition probabilities into the optimistic transition probabilities required for scheduling.

\begin{algorithm}[t] 
\SetAlgoLined
\SetAlgoNoEnd
\SetKwInOut{Input}{input}
\caption{\small Online Maximum Gain First (Online-MGF) Policy} \label{algo1}

Initialize $N$ users, constraint $M$, episode length $H$.\\
Initialize $B_n ^1 (c, \delta, c') = 0$ for all $c, \delta, c'$.\\
Initialize $\lambda^{(1)}$.\\
\For {episode $k=1,2,\ldots$}
{Reset $t=1$ and initial state $(\mathbf{c}, \boldsymbol{\delta})=(\mathbf{c} (1), \boldsymbol{\delta} (1))$.\\
Compute transition probabilities for each arm by solving the optimization problem \eqref{opt_value}-\eqref{opt_value_constraint}.\\
Compute gain indices from \eqref{gain_idx} for each arm at every time-slot $t$.\\
Schedule $M$ users with highest positive gain indices at every time-slot $t$.\\
Observe transitions and update visits $B_n ^k (c, \delta, c')$, empirical transition probability matrix $\hat {\mathbf{P}}^ {k}$, confidence region $\mathcal{B}^ {k}$.\\
Update $\lambda ^{k+1}$ from \eqref{lam_update}.}
\end{algorithm}

\subsubsection{Algorithm for Online-MGF Policy}
Using an optimistic approach, we estimate the transition probabilities \cite{orneeinfocom2026, wang2023optimistic, wang2024online}. 
We choose the optimistic transition probabilities $P_n ^k (c'|c, \delta)$ for each arm $n$ in episode $k$ that maximizes the value function within the confidence radius $\rho_n ^k (c, \delta)$. The optimization problem at episode $k$ in time-slot $t$ is given by
\begin{align}
&\max_{{P}_n ^k \in \mathcal{B} ^k} V_{n, t}^{P_n ^k} (c, \delta; \lambda^k), \label{opt_value}\\
&\quad \text{s.t.} ~{\thinspace} V_{n, t}^{P_n ^k} (c, \delta; \lambda^k) = \max_{a \in \mathcal{A}} Q_{n, t}^{P_n ^k} (c, \delta, a; \lambda^k), \label{opt_value_constraint}
\end{align}
where $Q_{n, t}^{P_n ^k} (c, \delta, a; \lambda^k)$ is the action-value function given by
\begin{align} \label{action_optimistic}
& Q_{n, t}^{P_n ^k} (c, \delta, a; \lambda^k) = \nonumber\\
&\begin{cases}
& \!\!\!\!\!\!\! \sum\limits_{c' \in \mathcal{C}} P_n ^k (1|c, \delta) (h (c', 1) + V_{n, t}^{P_n ^k} (c', 1; \lambda^k)) - \lambda^k,  \text{if} {\thinspace} a=1, \\
& \!\!\!\!\!\!\! V_{n, t}^{P_n ^k} (c, \delta+1; \lambda^k), ~~~~~~~~~~~~~~~~~~~~~~~~~~~~~~~~~~\text{if} {\thinspace} a=0. \\
\end{cases} 
\end{align}
After computing the optimistic transition probabilities and optimistic value functions, we utilize the optimistic action-value functions for each arm $n$ in episode $k$ to compute the gain index as follows:
\begin{align} \label{gain_idx}
\alpha_{n, t}^k (c, \delta; \lambda^k) = Q_{n, t}^{P_n ^k} (c, \delta, 1; \lambda^k) - Q_{n, t}^{P_n ^k} (c, \delta, 0; \lambda^k).    
\end{align}

The algorithm to compute the Online-MGF policy is provided in Algorithm \ref{algo1}. At every episode $k$, this policy selects $M$ arms with the highest positive gain indices \eqref{gain_idx}. 

\begin{remark}\label{remark2}
Existing Lagrangian-based online Whittle index policies require solving two sequential optimization problems: first, to estimate the optimistic transition probabilities and, second, to compute the Whittle indices \cite{wang2023optimistic}. Specifically, after estimating the optimistic transition probabilities, the algorithm in \cite{wang2023optimistic} requires solving an inverse root-finding problem to calculate the Whittle indices for every state. 
Unlike \cite{wang2023optimistic}, our Algorithm \ref{algo1} only requires solving one optimization problem \eqref{opt_value}-\eqref{opt_value_constraint} to obtain the optimistic transition probabilities $P_n ^k (c'|c, \delta)$ for the current CSI $C_n(t)=c'$ at time-slot $t$ given the outdated CSI $C_n(t-\Delta_n(t))=c$ and AoCSI $\Delta_n(t)=\delta$. 
Once the optimistic probabilities $P_n ^k (c'|c, \delta)$ are obtained, the gain index in \eqref{gain_idx} is computed using optimistic action-value functions \eqref{action_optimistic} from scalar subtraction. Hence, Algorithm \ref{algo1} does not need to solve the root-finding step as  required in \cite{wang2023optimistic}.
\end{remark}

\begin{remark} \label{remark3} \textbf{(Computational Complexity.)}
The Online-MGF policy in Algorithm \ref{algo1} reduces computational complexity compared to existing online learning approaches for RMABs \cite{wang2024online, wang2023optimistic}. In each episode, Algorithm \ref{algo1} computes the optimistic transition probabilities and action-value functions by using backward induction, which yields a complexity of $\mathcal{O}(N H \tau |\mathcal{C}|^2)$. Next, to make the scheduling decisions, Algorithm \ref{algo1} exhibits $\mathcal{O}(H N \log N)$ complexity for sorting. Hence, a total per-episode complexity of Algorithm \ref{algo1} is $\mathcal{O}(N H \tau |\mathcal{C}|^2 + H N \log N)$. In \cite{wang2024online}, the authors estimate the transition probability matrices for both active and passive actions, which yields the estimation complexity to scale quadratically with the state-space size as $\mathcal{O}(N \tau^2 |\mathcal{C}|^2)$. Furthermore, in \cite{wang2023optimistic}, the online Lagrangian-based Whittle index policies require an additional optimization step to calculate the Whittle indices for every state, introducing a complexity of $\mathcal{O}(N \tau^3 |\mathcal{C}|^3)$. By leveraging the deterministic behavior of passive AoCSI transitions and directly utilizing the gain indices, our Algorithm \ref{algo1} entirely eliminates the need of these computations associated with the transition matrix ($\mathcal{O}(\tau^2 |\mathcal{C}|^2)$) and Whittle indices ($\mathcal{O}(\tau^3 |\mathcal{C}|^3)$).
\end{remark}


The Lagrange multiplier $\lambda ^k$ is updated in every $k$-th episode by using the stochastic sub-gradient descent method \cite{nedic2008subgradient},\cite{avrachenkov2022whittle}. The update rule is given by
\begin{align}
&\lambda ^{k+1}= \max\bigg\{\lambda^k + \frac{\eta}{kH} \bigg(\sum_{n=1}^{N} \sum_{t=1}^{H} a_{n} ^{k} (t) -M\bigg), 0\bigg\}, \label{lam_update}
\end{align}
where $\eta/(kH)$ is the learning parameter, and action $a_{n} ^{k} (t)$ is obtained by solving the following decoupled problem:
\begin{align}
a_{n} ^{k} (t) = \arg\max_{a} Q_{n, t}^{P_n ^k} (c, \delta, a; \lambda_k).
\end{align}

\subsection{Regret Bound for Algorithm \ref{algo1}}\label{regret_analysis}

In this section, we analyze the regret of Algorithm \ref{algo1}. For clarity of analysis, we split the regret bound \eqref{reg_new} in Definition \ref{def1} into two components:
\begin{align}
& \text{Reg}_{\text{comp}_1}= \sum_{n=1}^{N} L_n({\pi_n^{*}}, P_n, \lambda^*) - \sum_{n=1}^{N} L_n(\pi_n^{k}, P_n, \lambda^*), \\
& \text{Reg}_{\text{comp}_2}=  \sum_{n=1}^{N} L_n(\pi_n^{k}, P_n, \lambda^*) - J^k (\pi^k, \mathbf{P}).
\end{align}
Specifically, $\text{comp}_1$ captures the performance degradation due to estimating the unknown channel transition probabilities. In contrast, $\text{comp}_2$ measures the sub-optimality gap between the theoretical upper bound (the relaxed Lagrangian problem) and our actual MGF  policy. In the following analysis, we first bound the regret associated with $\text{comp}_1$ and then we bound $\text{comp}_2$. Finally, we combine these two bounds to establish the total cumulative regret $\text{Reg}(K)$ of the proposed Online-MGF policy.

To bound the regret, we first show that with high probability the true transition $\mathbf{P}$ lies within the confidence region $\mathcal{B} ^{k}$ in \eqref{ball_p}.

\begin{lemma} \label{lemma1}
Given $\mu > 0$ and $k \geq 1$, we get that $\emph{Pr}(\mathbf{P} \in \mathcal{B} ^{k}) \geq 1-\frac{\mu}{k^2}$.
\end{lemma}


\begin{proof}
See Appendix \ref{proof_lemma1}.
\end{proof}


Lemma \ref{lemma1} implies that for each episode $k$, a confidence region can be obtained within which the actual transition probability lies with high probability. Utilizing Lemma \ref{lemma1}, we bound the regret associated with $\text{comp}_1$ for each episode $k$ in the following lemma.



\begin{lemma} \label{lemma2}
For given initial state $((c_{1} (1), \delta_{1} (1)), (c_{2} (1), \delta_{2} (1)), \ldots, (c_{N} (1), \delta_{N} (1)))$, the following holds with probability $1-\mu$: 
\emph{\begin{align} \label{lemma2_3}
    & \mathrm{Reg}_{\text{comp}_1} (k) \nonumber\\
    \leq & \sum_{n=1}^{N} V_{n, 1} ^{P_n ^k} (c_{n} (1), \delta_{n} (1); \lambda^k) \!-\! V_{n, 1} ^{P_n} (c_{n} (1), \delta_{n} (1); \lambda^k). 
\end{align}}

\end{lemma}

\begin{proof}
See Appendix \ref{proof_lemma2}.
\end{proof}
Next, we utilize Lemma \ref{lemma2} to bound the regret associated with $\text{comp}_1$ for $K$ episodes. In this sequel, we have the following theorem.



\begin{theorem} \label{theorem_regret}
With probability $1-\mu$, the cumulative regret of Algorithm \ref{algo1} associated with \emph{$\text{comp}_1$} in all $K$ episodes is given by
\emph{\begin{align} \label{theorem_regret_eq1}
\text{Reg}_{\text{comp}_1}(K)  \leq \mathcal{O} \bigg( V_{\max} \Gamma |\mathcal{C}| N H \sqrt{K \log K}\bigg),
\end{align}}
\!\!where $V_{\max}$ is the maximum value-function, $\Gamma$ is the maximum value of AoCSI, $|\mathcal{C}|$ is the size of the state space of CSI, $N$ is the number of users, $H$ is the time-horizon length in each episode, and $K$ is the total number of episodes.
\end{theorem}




\begin{proof}
See Appendix \ref{proof_theorem_regret}.
\end{proof}

Next, we will find the regret associated with $\text{comp}_2$. Define
\begin{align} \label{g}
g(N) = \sum_{n=1}^{N} L_n(\pi_n ^k, P_n ^k, \lambda^*) - J^k (\pi^k, \mathbf{P}^k).
\end{align}
The function $g (N)$ represents the gap between the performance of the Lagrangian problem under its optimal solution and the main problem under the MGF policy. In $g (N)$, both the optimal solution of Lagrangian problem and the MGF policy are associated with transition probability $\mathbf{P}^k$. Hence, $g (N)$ does not represent the regret, but instead it measures the optimality gap of the MGF policy even if we know the actual transition probabilities. 

In this sequel, we have the following lemma.
\begin{lemma} \label{lemma_g}
The sub-optimality gap \eqref{g} between the Lagrangian dual problem and the original problem is upper-bounded as follows
\begin{align}
g(N) \leq O \bigg(\frac{K}{\sqrt{N}}\bigg),
\end{align}
where $N$ is the number of users and $K$ is the number of episodes.
\end{lemma}

\begin{proof}
See Appendix \ref{proof_lemma_g}.
\end{proof}

An immediate corollary of Lemma \ref{lemma_g} is
\begin{corollary}
If the number of users $N \to \infty$, then $g (N) \to 0$.
\end{corollary}


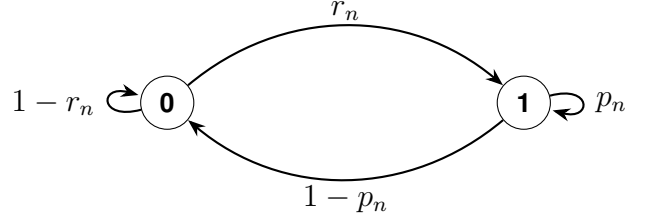
\begin{figure}[t]
    \centering
    \begin{tikzpicture}[
        state/.style={
            circle,        
            draw,          
            minimum size=2em, 
            inner sep=0pt, 
            font=\sffamily\bfseries 
        },
        every edge/.style={
            draw,          
            ->,            
            >=Stealth,     
            thick          
        },
        prob_label/.style={
            font=\large, 
            inner sep=1pt,      
            xshift=0em, yshift=0em 
        }
    ]

        \node[state] (s0) {0};
        \node[state, right=4cm of s0] (s1) {1};


        \path (s0) edge[
            loop left,
            "{$1-r_n$}" {prob_label, xshift=-0.4em, yshift=0.0em} 
        ] (s0);

        \path (s0) edge[
            bend left=40, 
            "{$r_n$}" {prob_label, yshift=0.1em} 
        ] (s1);

        \path (s1) edge[
            loop right,
            "{$p_n$}" {prob_label, xshift=0.4em, yshift=0em} 
        ] (s1);

        \path (s1) edge[
            bend left=40, 
            "{$1-p_n$}" {prob_label, yshift=-0.1em} 
        ] (s0);

    \end{tikzpicture}
    \vspace{1em}
    \caption{Two-state Markov chain model for channel of user $n$.}
    \label{MC}
\end{figure}

\begin{theorem} \label{theorem_regret_2}
If $N \to \infty$, then with probability $1-\mu$, the cumulative regret of Algorithm \ref{algo1} associated with \emph{$\text{comp}_2$} in all $K$ episodes is given by
\begin{align} \label{theorem_regret_eq1}
\emph{Reg}_{\text{comp}_2}(K)  \leq \mathcal{O} \bigg(2 V_{\max} \Gamma |\mathcal{C}| N H \sqrt{K \log K}\bigg),
\end{align}
where 
$V_{\max}$ is the maximum value-function, $\Gamma$ is the maximum value of AoCSI, $|\mathcal{C}|$ is the size of the state space of CSI, $N$ is the number of users, $H$ is the time-horizon length in each episode, and $K$ is the total number of episodes.
\end{theorem}

\begin{proof}
See Appendix \ref{proof_theorem_regret_2}.
\end{proof}

\begin{theorem} \label{theorem_regret_3}
If $N \to \infty$, then with probability $1-\mu$, the cumulative regret of Algorithm \ref{algo1} in all $K$ episodes is given by
\begin{align} \label{theorem_regret_eq1}
\emph{Reg} (K) \leq \mathcal{O} \bigg(3 V_{\max} \Gamma |\mathcal{C}| N H \sqrt{K \log K}\bigg),
\end{align}
where 
$V_{\max}$ is the maximum value-function, $\Gamma$ is the maximum value of AoCSI, $|\mathcal{C}|$ is the size of the state space of CSI, $N$ is the number of users, $H$ is the time-horizon length in each episode, and $K$ is the total number of episodes.
\end{theorem}

\begin{proof}
See Appendix \ref{proof_theorem3}.
\end{proof}

Theorem \ref{theorem_regret_3} illustrates that Algorithm \ref{algo1} achieves sub-linear regret with respect to the number of episodes $K$ as $N \to \infty$. This regret bound aligns with existing online learning algorithms in literature \cite{wang2023optimistic, shisher2025onlinelearningwhittleindices, tripathi2021onlinelearningapproachoptimizing,  akbarzadeh2023learning, neu2013efficient, foster2020beyond}. 
Because we generalize our problem to a multiple channel state scenario, the analysis cannot be simplified as our previous study with ON/OFF channel states \cite{orneeinfocom2026}. 
In section \ref{sim}, our numerical results illustrate that the Online-MGF policy converges to both the MGF policy with known 
channel statistics within a small number of episodes.

\section{Special Case: ON/OFF Channel}

Consider a simplified scenario where the CSI $C_n (t)$ is modeled as a discrete-time, two-state time-homogeneous Markov chain (ON/OFF channels, where $C_n (t) \in \{0, 1\} $). 
If $C_n(t)=0$, the channel of user $n$ is in the OFF state; otherwise, if $C_n(t)=1$, the channel of user $n$ is in the ON state. 
The state transition of the channels for user $n$ is provided in Figure \ref{MC} and the corresponding transition matrix is given by
\begin{align}\label{transitionMatrix}
P_n=
\begin{bmatrix}
P_n (1|1) & P_n (0|1)\\
P_n (1|0) & P_n (0|0)
\end{bmatrix}=
\begin{bmatrix}
p_n & 1-p_n\\
r_n & 1-r_n
\end{bmatrix},
\end{align}
where 
\begin{align}
& p_n = P_n (C_n (t) =1 |C_n (t-1) =1), \label{p_m}\\
& r_n = P_n (C_n (t) =1 |C_n (t-1) =0). \label{r_m}
\end{align}

\begin{algorithm}[t] 
\SetAlgoLined
\SetAlgoNoEnd
\SetKwInOut{Input}{input}
\caption{\small Online Maximum Gain First (Online-MGF) Policy for ON/OFF channel} \label{algo2}

Initialize $N$ users, constraint $M$, episode length $H$.\\
Initialize $B_n ^1 (c, \delta, c') = 0$ for all $c, \delta, c'$.\\
Initialize $\lambda^{(1)}$.\\
\For {episode $k=1,2,\ldots$}
{Reset $t=1$ and initial state $(\mathbf{c}, \boldsymbol{\delta})=(\mathbf{c} (1), \boldsymbol{\delta} (1))$.\\
Compute 
$P_n ^k (1|c, \delta)$ for each arm $n$ and for all $(c, \delta)$ by solving \eqref{opt_value}-\eqref{opt_value_constraint}.\\
Compute gain indices from \eqref{gain_idx} for each arm at every time-slot $t$.\\
Schedule $M$ users with highest gain indices at every time-slot $t$.\\
Observe transitions and update visits $B_n ^k (c, \delta, c')$, empirical transition probability $\hat {\mathbf{P}}^ {k}$, confidence region $\mathcal{B}^ {k}$.\\
Update $\lambda ^{k+1}$ from \eqref{lam_update}.}
\end{algorithm}

For this simplified scenario, the function $f(c, \delta, a) $ can be written as
\begin{align} \label{f_two_state}
& f (C_n (t-\Delta_n (t)), \Delta_n (t), a_n (t)) \nonumber\\
=& a_n (t) \mathbb{E} [C_n (t)| C_n (t-\Delta_n (t)), \Delta_n (t)],
\end{align}
where we consider linear throughput function $h (C_n (t), a_n (t)) = C_n (t) a_n (t) $.
If $a_n (t)=0$, then $f_n (C_n (t-\Delta_n (t)), \Delta_n (t), 0)=0$; If $a_n (t)=1$, then $f_n (C_n (t-\Delta_n (t)), \Delta_n (t), 1)$ can be further simplified: Given $C_n (t-\Delta_n (t))=c$ and $\Delta_n(t)=\delta$, we have
\begin{align} \label{func_f_a_binary}
f (c, \delta, 1) = P_n(1|c, \delta),
\end{align}
 where $P_n(1|c, \delta)$ can be calculated as
\begin{align}
& P_n(1|0, \delta) = \frac{r_n (1 - (p_n -r_n)^{\delta})}{1-p_n +r_n}, \label{P10}\\
&  P_n(1|1, \delta) = \frac{r_n + (1-p_n) (1-(p_n - r_n)^{\delta})}{1-p_n +r_n}. \label{P11}
\end{align}
Becuase of this simplification, we need to compute only the transition probabilities $P_n (1|c, \delta)$ in the online learning algorithm. Using \eqref{func_f_a_binary}, we can simplify the action-value function \eqref{ac_value} associated with ON/OFF channels.
\begin{align} \label{value_simple_binary}
& Q_{n, t} (c, \delta, a; \lambda) = \nonumber\\
& \begin{cases}
& \!\!\!\!\!\!\! P_n (1|c, \delta) \!+\!\sum_{c' \in \{0, 1\}} P_n (c'|c, \delta) V_{n, t+1} (c', 1; \lambda) - \lambda,\\
&~~~~~~~~~~~~~~~~~~~~~~~~~~~~~~~~~~~~~~~~~~~~~~~~~~~~~~\text{if} {\thinspace} a=1, \\
& \!\!\!\!\!\!\!  V_{n, t+1} (c, \delta+1; \lambda), ~~~~~~~~~~~~~~~~~~~~~~~~~~~~~~~~~~~\text{if} {\thinspace} a=0.\\
\end{cases}
\end{align}
To that end, the Online-MGF policy in Algorithm \ref{algo2} only requires solving one optimization problem \eqref{opt_value}-\eqref{opt_value_constraint} to get probabilities $P_n ^k (1|c, \delta)$ for the CSI $C_n(t)=1$ at time-slot $t$ given outdated CSI $C_n(t-\Delta_n(t))=c$ and AoCSI $\Delta_n(t)=\delta$. This also yields $P_n ^k (0|c, \delta)=1-P_n ^k (1|c, \delta)$. Therefore, we do not need to compute the entire transition probabilities from state $(c, \delta)$ to next state $(c', \delta')$. 
The AoCSI evolution is known since whenever source $n$ is scheduled AoCSI becomes $1$, and otherwise, it increases by $1$.

\begin{algorithm}[t] 
\SetAlgoLined
\SetAlgoNoEnd
\SetKwInOut{Input}{input}
\caption{\small Online Whittle Index Policy for ON/OFF channel} \label{algo_3}

Initialize $N$ users, constraint $M$, episode length $H$.\\
Initialize $B_n ^1 (c, \delta, c') = 0$ for all $c, \delta, c'$.\\
\For {episode $k=1,2,\ldots$}
{Reset $t=1$ and initial state $(\mathbf{c}, \boldsymbol{\delta})=(\mathbf{c} (1), \boldsymbol{\delta} (1))$.\\
Compute 
$P_n ^k (1|c, \delta)$ for each arm $n$ and for all $(c, \delta)$ by solving \eqref{opt_value}-\eqref{opt_value_constraint}.\\
Compute Whittle indices from \eqref{online_whittle} for each arm at every time-slot $t$.\\
Schedule $M$ users with highest Whittle indices at every time-slot $t$.\\
Observe transitions and update visits $B_n ^k (c, \delta, c')$, empirical transition probability $\hat {\mathbf{P}}^ {k}$, confidence region $\mathcal{B}^ {k}$.\\}
\end{algorithm}

\subsection{Whittle Index Policy}

Another benefit of this model is that we can prove indexability and derive a closed-form expression of the Whittle index for the infinite-horizon average-cost setting. Our problem \eqref{problem_simple}-\eqref{constraint_simple} considers a finite horizon $T$. However, computing time-dependent Whittle indices is computationally intractable and generally does not yield a closed-form
solution. Therefore, following standard practice in RMAB literature \cite{whittle_restless, weber1990index}, we analyze the infinite-horizon case where the value functions and action-value functions become time-invariant. This allows us to develop a Online Whittle index policy. 

To that end, we find the Whittle index policy for the following problem
\begin{align}
& \lim_{T \to \infty} \frac{1}{T} \max_{\pi \in \Pi} \sum_{t=1}^{T} \sum_{n=1}^{N} \mathbb{E}_{\pi} \!\bigg[ \!f (C_n (t\!-\!\Delta_n (t)), \Delta_n (t), a_n (t))\!\bigg] \label{problem_simple_w} \\
&~ \text{s.t.} ~ \sum_{n=1}^{N} a_n (t) \leq M,~t=1,2,\ldots,T. \label{constraint_simple_w}
\end{align}


Define $\Phi (\lambda)$ as the set of states $(c, \delta) \in \{0, 1\} \times \mathbb{N}$ such that if $C_n (t-\Delta_n (t)) = c$ and $\Delta_n (t) = \delta$, the optimal solution for \eqref{decoupled} is to take a passive action at time $t$.

\begin{figure*}[t]
  \centering 
  \begin{subfigure}[t]{0.32\textwidth}
\includegraphics[width=\textwidth]{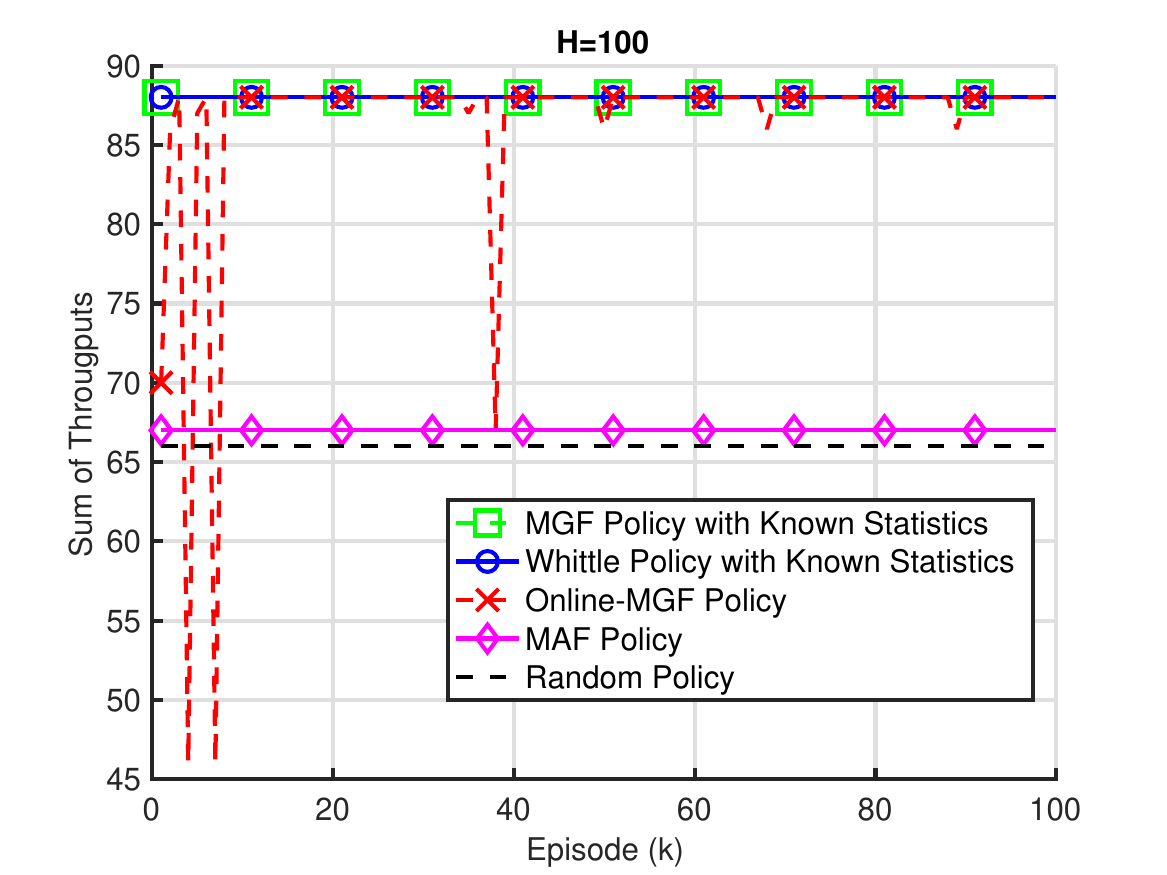}
\end{subfigure}
%
\begin{subfigure}[t]{0.32\textwidth}
\includegraphics[width=\textwidth]{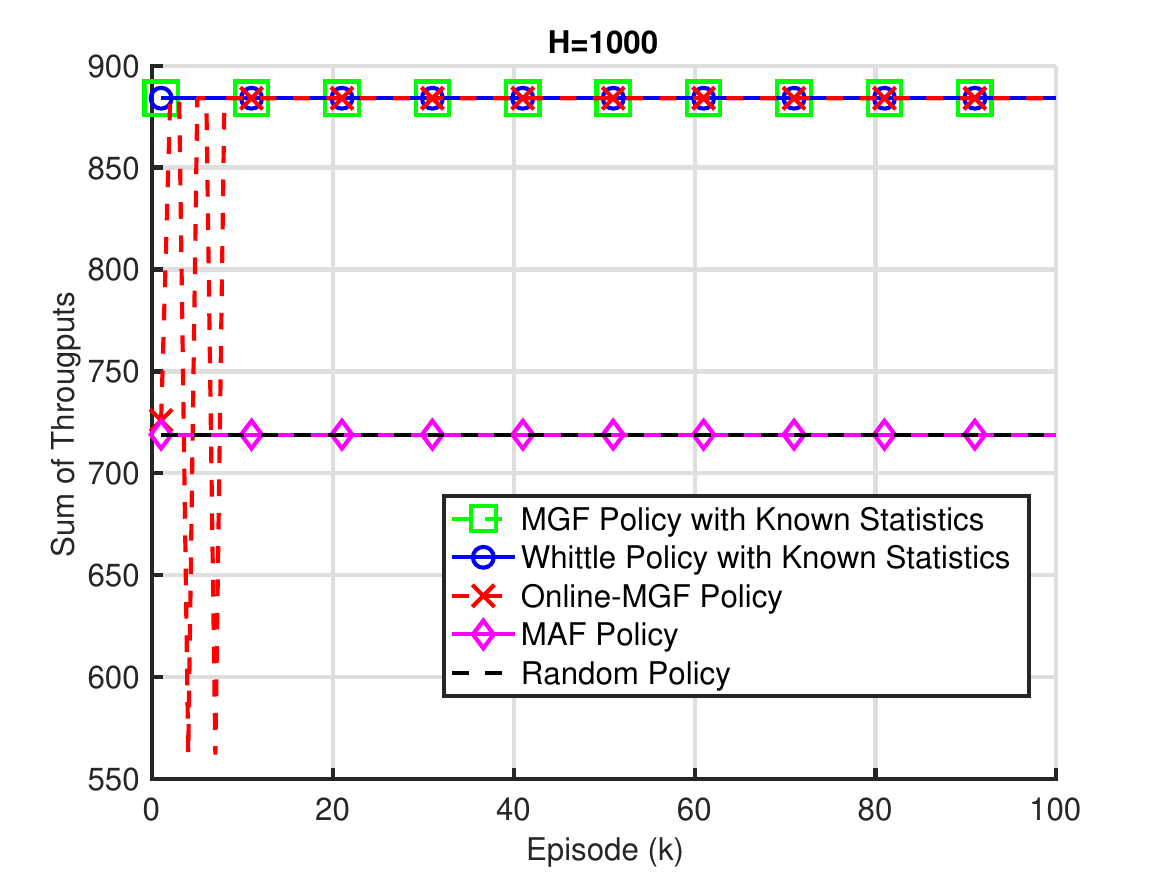}
\end{subfigure}
%
\begin{subfigure}[t]{0.32\textwidth}
\includegraphics[width=\textwidth]{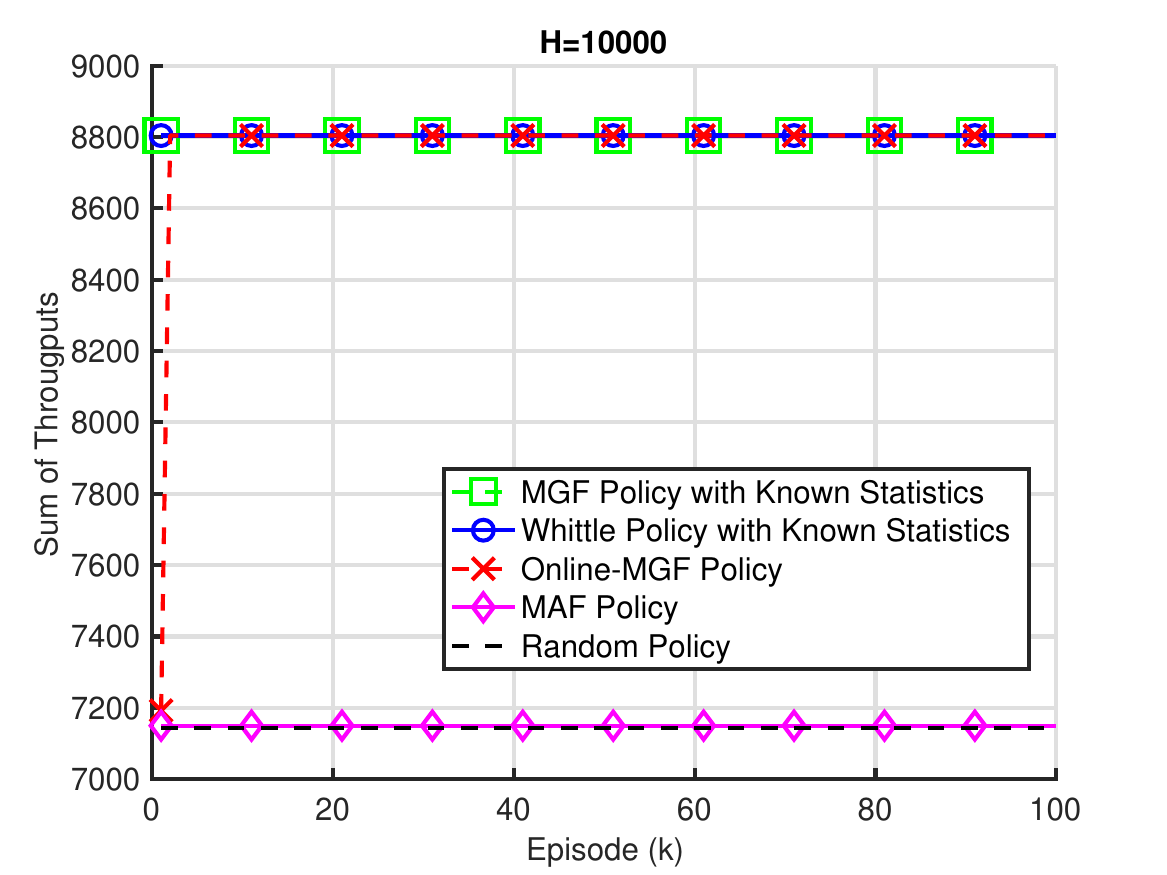}
\end{subfigure}
\vspace{0.05in}
\caption{\small Sum of Throughputs vs the number of episodes ($k$) with different time-horizon length in each episode ($H$). The other parameters are  $p_1=0.6$, $r_1=0.5$, $p_2=0.9$, $r_2=0.75$, $N=4$, and $M=1$.\label{fig:simulationH}}
\end{figure*}

\begin{figure*}[t]
  \centering 
  \begin{subfigure}[t]{0.32\textwidth}
\includegraphics[width=\textwidth]{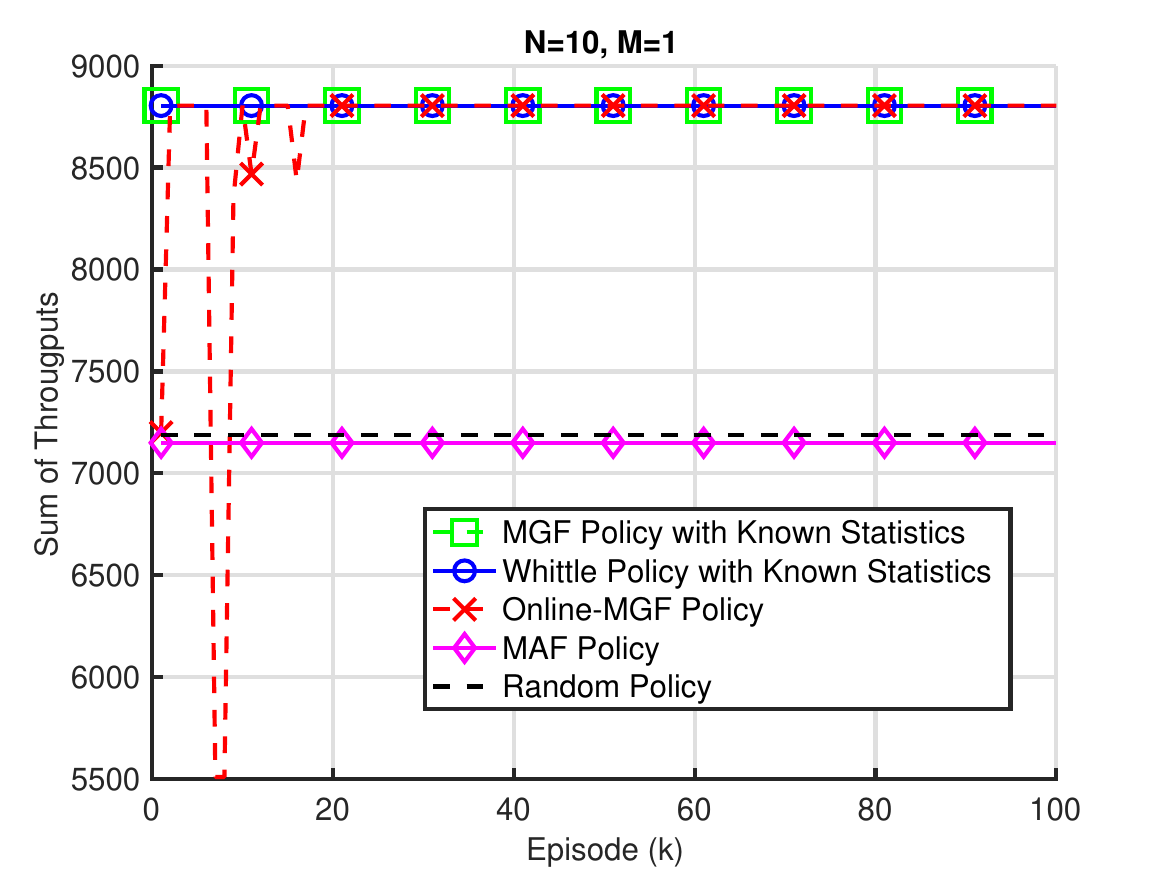}
\end{subfigure}
%
\begin{subfigure}[t]{0.32\textwidth}
\includegraphics[width=\textwidth]{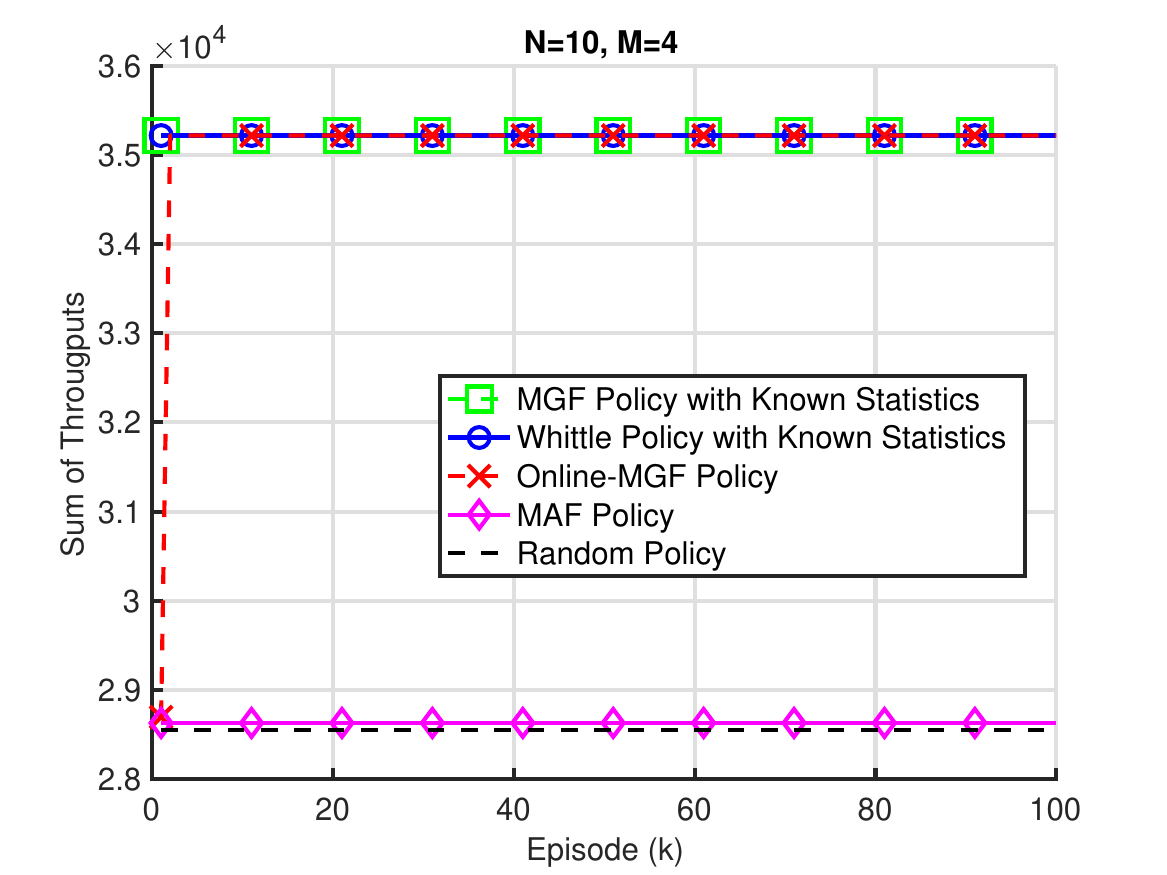}
\end{subfigure}
%
\begin{subfigure}[t]{0.32\textwidth}
\includegraphics[width=\textwidth]{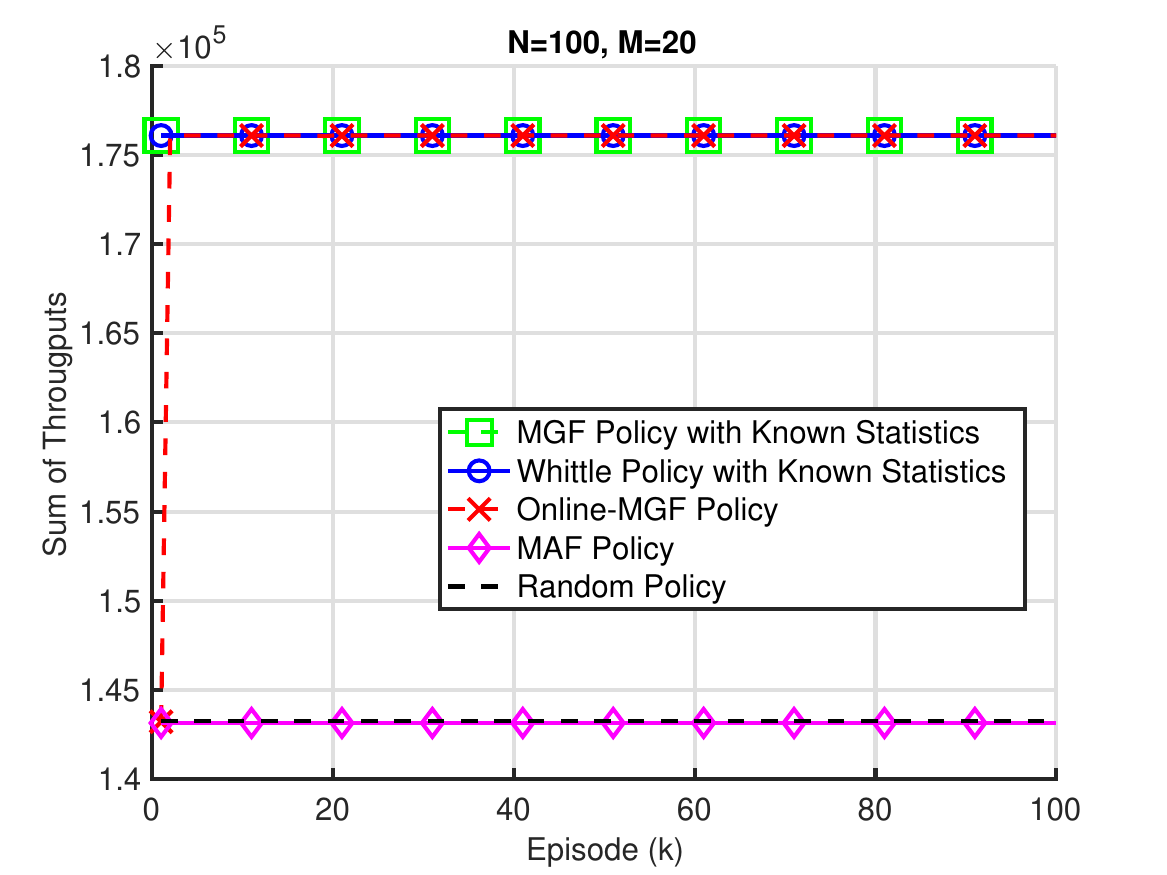}
\end{subfigure}
\vspace{0.05in}
\caption{\small Sum of Throughputs vs the number of episodes ($k$) with different configurations of $N$ and $M$. The other parameters are $p_1=0.6$, $r_1=0.5$, $p_2=0.9$, $r_2=0.75$, and $H=10000$.  \label{fig:simulationN}}
\end{figure*}

\begin{figure*}[t]
  \centering 
  \begin{subfigure}[t]{0.32\textwidth}
\includegraphics[width=\textwidth]{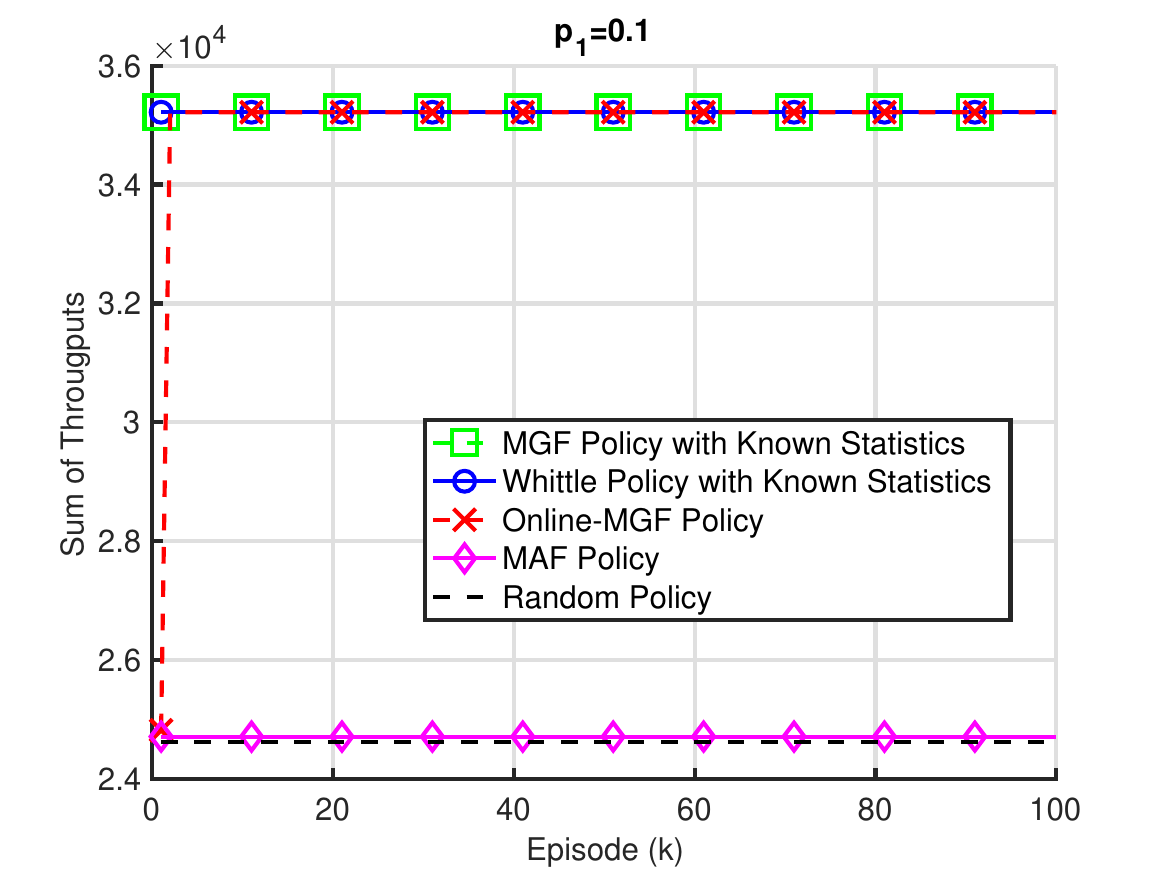}
\end{subfigure}
%
\begin{subfigure}[t]{0.32\textwidth}
\includegraphics[width=\textwidth]{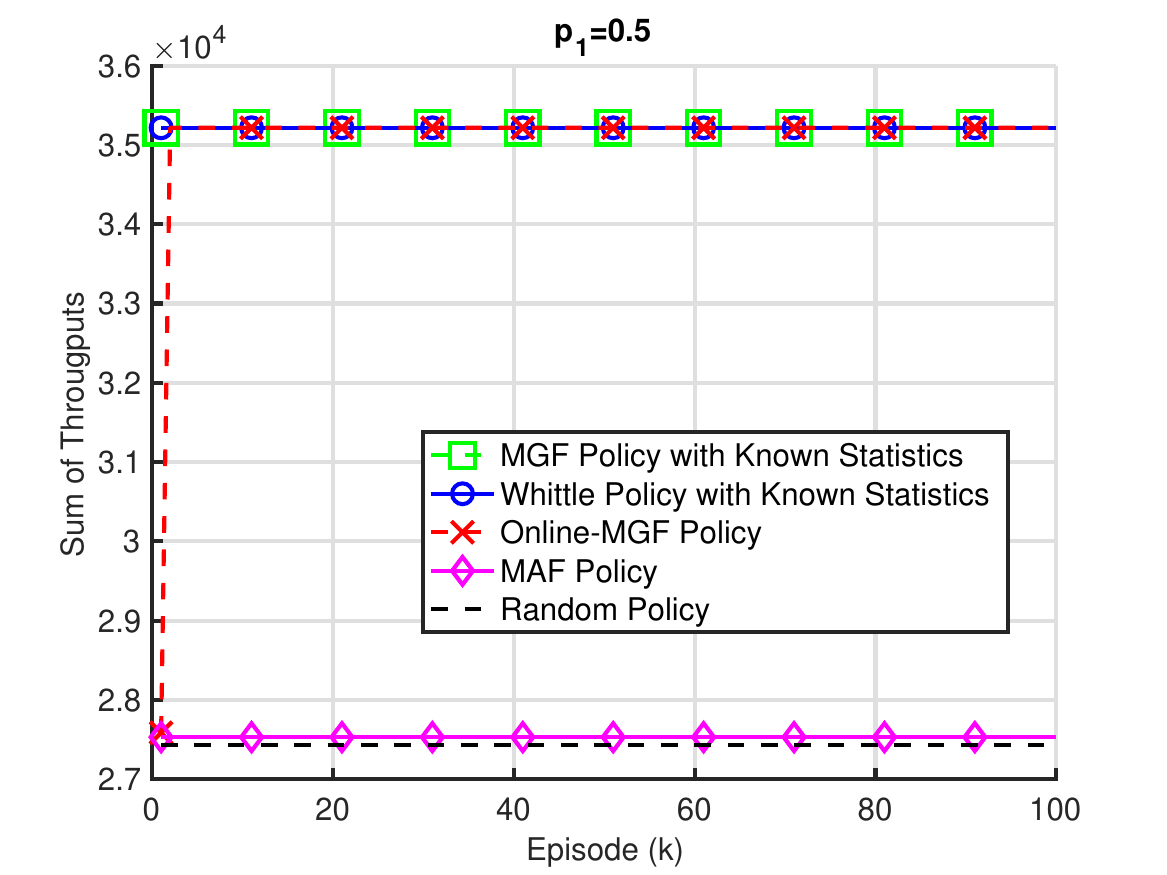}
\end{subfigure}
%
\begin{subfigure}[t]{0.32\textwidth}
\includegraphics[width=\textwidth]{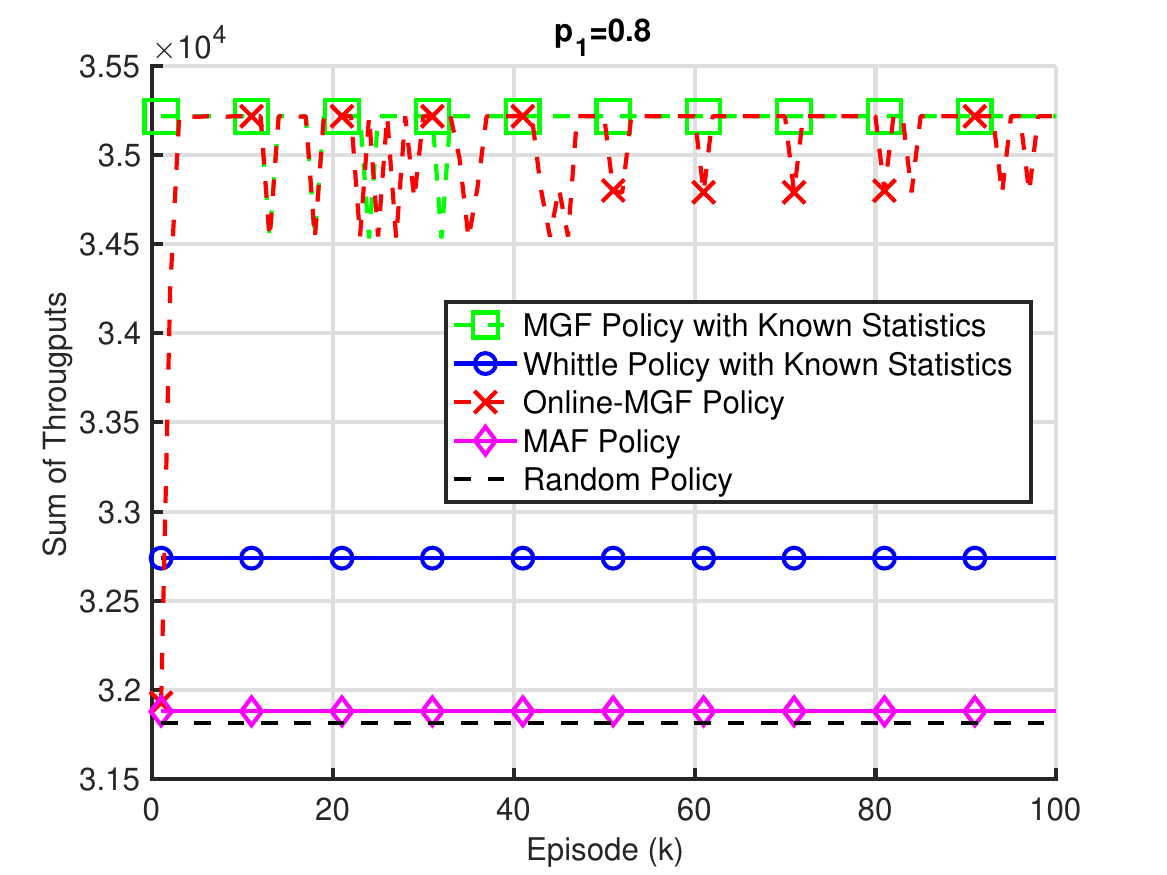}
\end{subfigure}
\vspace{0.05in}
\caption{\small Sum of Throughputs vs the number of episodes ($k$) with different channel statistics. The other parameters are $H=10000, N=10$, and $M=4$.  \label{fig:simulationP}}
\end{figure*}

\begin{figure}
\centering
\includegraphics[width=8cm]{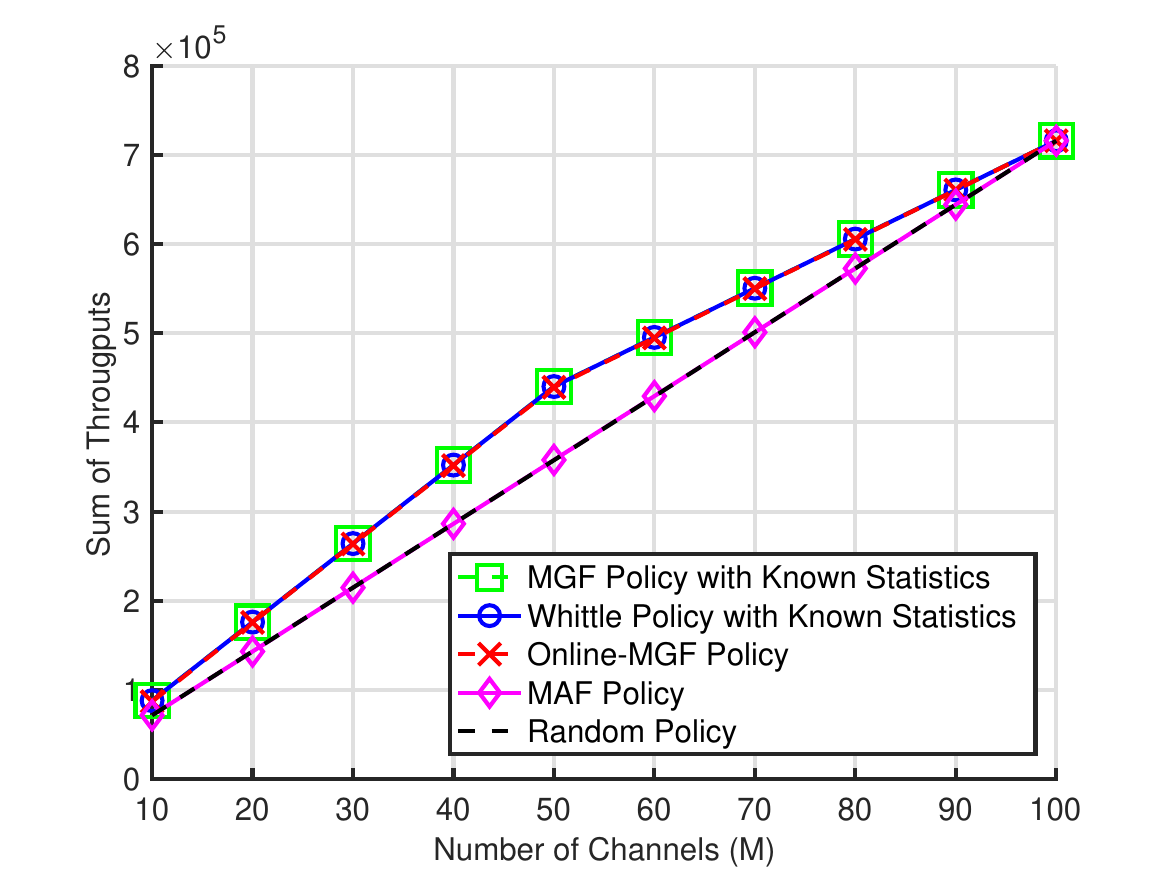}
\caption{Sum of Throughputs vs Number of Channels ($M$). The other parameters are $p_1=0.6$, $r_1=0.5$, $p_2=0.9$, $r_2=0.75$, $H=10000$, and $N=100$.}\vspace{-0.0cm}
\label{fig:simulationM}
\end{figure}

\begin{definition} \label{indexability_def}
\textbf{(Indexability).} 
\cite{verloop2016asymptotically} An arm is said to be \emph{indexable} if, as the activation cost $\lambda$ increases from $-\infty$ to $\infty$, the set $\Phi (\lambda)$ increases monotonically, i.e., $\lambda_1 \leq \lambda_2$ implies $\Phi (\lambda_1) \subseteq \Phi (\lambda_2)$.
\end{definition}


In general, establishing indexability of an RMAB is a challenging task. However, we are able to prove indexability of problem \eqref{problem_simple}-\eqref{constraint_simple} for ON/OFF channels. This becomes possible because the reward function $f (c, \delta, 1)$ is greatly simplified in this case and only requires us to know the probability $P_n (1|c, \delta)$.

\begin{theorem} \label{theorem_5}
If the channel state $C_n (t)$ evolves as a two-state Markov chain, then, the RMAB problem \eqref{problem_simple}-\eqref{constraint_simple} is indexable.
\end{theorem}

\begin{proof}
See Appendix \ref{indexability_proof}.
\end{proof}

Next, we introduce the definition of the Whittle index. 
 
\begin{definition} \label{def_6}
\cite{whittle1988restless} If an arm is indexable, then the Whittle index $W_n (c, \delta)$ of the arm at state $(c, \delta)$ is defined by
\begin{align} \label{Whittle_index_general_definition}
W_n (c, \delta) = \inf_{\lambda \in \mathbb R} \{\lambda \in \mathbb{R} : (c, \delta) \in \Phi (\lambda)\},
\end{align}
which is the infimum of the activation cost $\lambda$ for which it is better not to activate the arm.
\end{definition}

\begin{theorem} \label{theorem_6}
The Whittle index $W_n (c, \delta)$ of problem \eqref{decoupled} at state $(c, \delta)$ is given by
\begin{align} \label{whittle_index}
W_n (c, \delta) = \begin{cases} \displaystyle \frac{P_n(1|1, \delta)}{1 - \big( P_n (1|1, 1) - P_n (1|1, \delta) \big)}, & \text{if } c = 1, \\ \\ 0, & \text{if } c = 0, \end{cases}
\end{align}
where $P_n (1|c, \delta)$ is defined in \eqref{P10}-\eqref{P11}.
\end{theorem}

\begin{proof}
See Appendix \ref{proof_Whittle}.
\end{proof}

The benefit of developing the Whittle index policy is that the algorithm does not need to solve $\lambda$ by stochastic sub-gradient descent method in every episode. Because the scheduler is unaware of the actual transition probability, the Whittle index at episode $k$ is given by
\begin{align} \label{online_whittle}
W_n ^k (c, \delta) = \begin{cases} \displaystyle \frac{P_n ^k (1|1, \delta)}{1 - \big( P_n ^k (1|1, 1) - P_n ^k (1|1, \delta) \big)}, & \text{if } c = 1, \\ \\ 0, & \text{if } c = 0, \end{cases}    
\end{align}
where $W_n ^k (c, \delta)$ represents the Whittle index at episode $k$.
The Online Whittle index policy is provided in Algorithm \ref{algo_3}. This algorithm selects $M$ users with the highest Whittle indices at every time-slot $t$.

\section{Simulation Results}\label{sim}
In this section, we evaluate the performance of the following policies:

\begin{itemize}

\item {\bf Random Policy}: This policy randomly selects users following a uniform distribution. In the simulation, we use MATLAB built-in function $\mathrm{randperm}(N, M)$. 

\item {\bf MGF Policy with Known Channel Statistic}: This policy computes the gain index using the known channel statistics. Then, this policy schedules $M$ users with the highest gain indices in every time-slot. 

\item {\bf Whittle Index Policy with Known Channel Statistic}: This policy computes the Whittle index using the known channel statistics. The closed-form expression for the Whittle index with known channel statistics of this problem is provided in \cite{Ouyang_TMC, Liu_TIT}. Then, the policy schedules $M$ users with the highest Whittle indices in every time-slot $t$. 

\item {\bf Online-MGF Policy}: The policy is provided in Algorithm \ref{algo1}. The Online-MGF policy does not know the channel statistics. This policy first learns the channel statistics, then computes gain indices at the beginning of each episode, and schedule $M$ users with the highest gain indices. 

\item {\bf Maximum Age First (MAF) policy}: This policy schedules $M$ users with the highest AoCSI values at every time-slot. 
\end{itemize}

\subsection{ON/OFF Channel}

In the simulation, we consider two classes of users, namely class $1$ and class $2$. Users in the same class have the same channel statistics. We consider that half of the users belong to class $1$ and another half belong to class $2$. We denote $P_n (C_n (t) =1 |C_n (t-1) =1)$ of user $n$ by $p_1$ and  $P_n (C_n (t) =1 |C_n (t-1) =0)$ by $r_1$ if the user $n$ belongs to class $1$; otherwise we denote $P_n (C_n (t) =1 |C_n (t-1) =1)$ of user $n$ by $p_2$ and  $P_n (C_n (t) =1 |C_n (t-1) =0)$ by $r_2$ if user $n$ belongs to class $2$.

Figure \ref{fig:simulationH} illustrates the sum of throughput vs the number of episodes ($k$) with different time horizon length in each episode $H=100, 1000, 10000$. In this figure, we set $p_1=0.6$, $r_1=0.5$, $p_2=0.9$, $r_2=0.75$, $N=4$, and $M=1$. According to Figure \ref{fig:simulationH}, MGF policy with known channel statistics shows similar performance as Whittle index policy with known channel statistics. The figure also depicts that the proposed Online-MGF policy converges faster to the MGF policy with known channel statistics as the number of time horizon $H$ in each episode increases. This is because for higher values of $H$, Online-MGF observes more state transitions in each episode, which yields accurate learning of channel statistics. Moreover, Online-MGF performs better compared to the other baselines, such as MAF and random policies. This is because MAF policy only considers AoCSI and ignores most recently observed CSI, whereas Random policy schedules the users randomly. 

Figure \ref{fig:simulationN} illustrates the sum of throughput vs number of episodes ($k$) with different configurations of $N$ and $M$. In this figure, we set $H=10000$ and the other parameters are the same as Figure \ref{fig:simulationH}. 
This figure shows that for all configurations, the Online-MGF policy converges faster to the MGF and the Whittle index policies with known channel statistics. As expected, the proposed Online-MGF policy performs better compared to MAF and random policies. 

Figure \ref{fig:simulationP} depicts the sum of throughput vs the number of episodes ($k$) with different values of $p_1=0.1, 0.5, 0.8$. In this figure, we set $N=10$, $M=4$ and other parameters are the same as Figure  \ref{fig:simulationN}. 
As expected, the Online-MGF policy outperforms the MAF and random policies.

In Figure \ref{fig:simulationM}, we plot the sum of throughput vs the number of channels ($M$). In this figure, the sum of throughput for each point is averaged over $100$ episodes. For this figure, we set $N=100$ and other parameters are the same as Figure \ref{fig:simulationN}. As the number of channels increases, the sum of throughput increases. The figure illustrates that the performance gain of Online-MGF policy with MAF and Random policies first increases with $M$ up to $M=50$. After $M=50$, the performance gain decreases because the number of available channels for allocating to $N=100$ users increases. 

Our simulation illustrates that the performance of our proposed Online-MGF policy aligns with the MGF policy and Whittle index policy with known channel statistics. Furthermore, Online-MGF policy outperforms MAF and random policy as the number of episodes increases in every settings. Notably, for $p_1 = 0.8$, the Online-MGF policy outperforms the Whittle index policy with known channel statistics.

\begin{figure*}[t]
 \centering 
%
\begin{subfigure}[t]{0.36\textwidth}
\includegraphics[width=\textwidth]{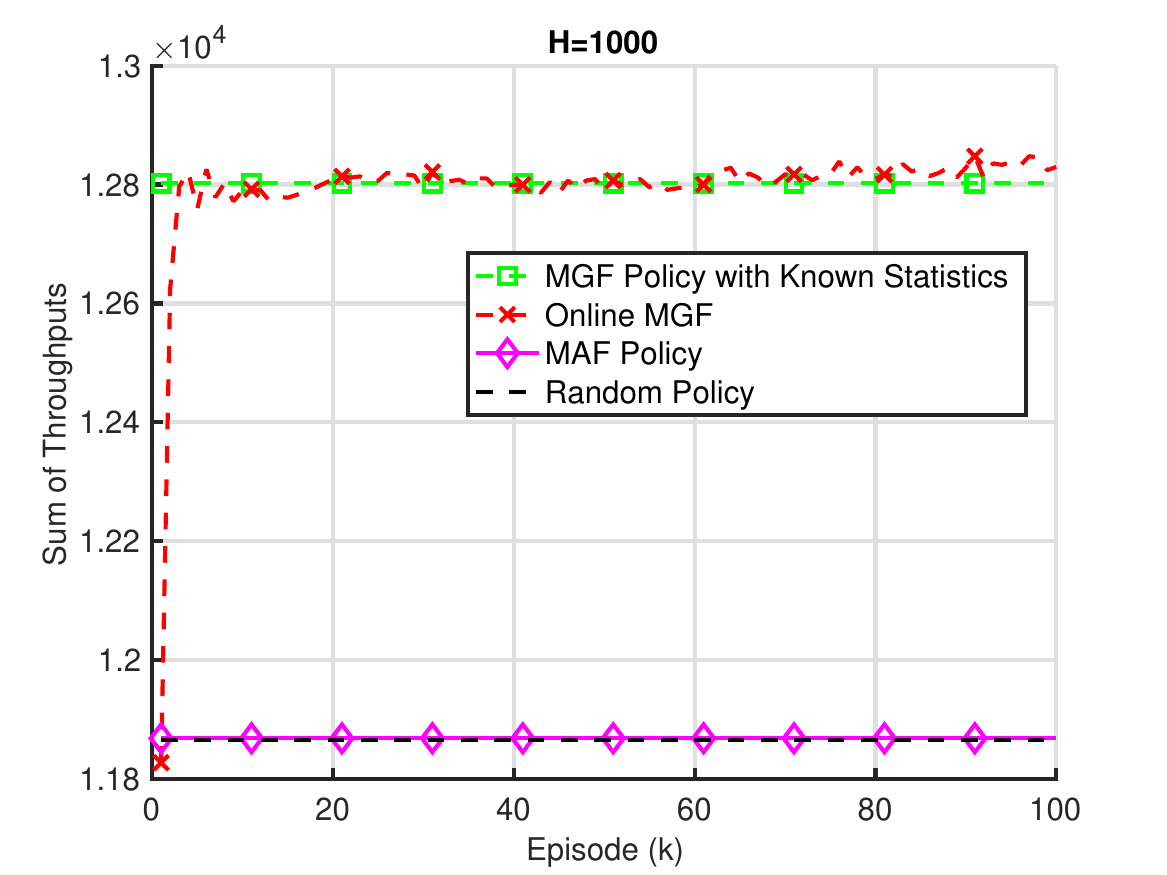}
\end{subfigure}
%
\begin{subfigure}[t]{0.36\textwidth}
\includegraphics[width=\textwidth]{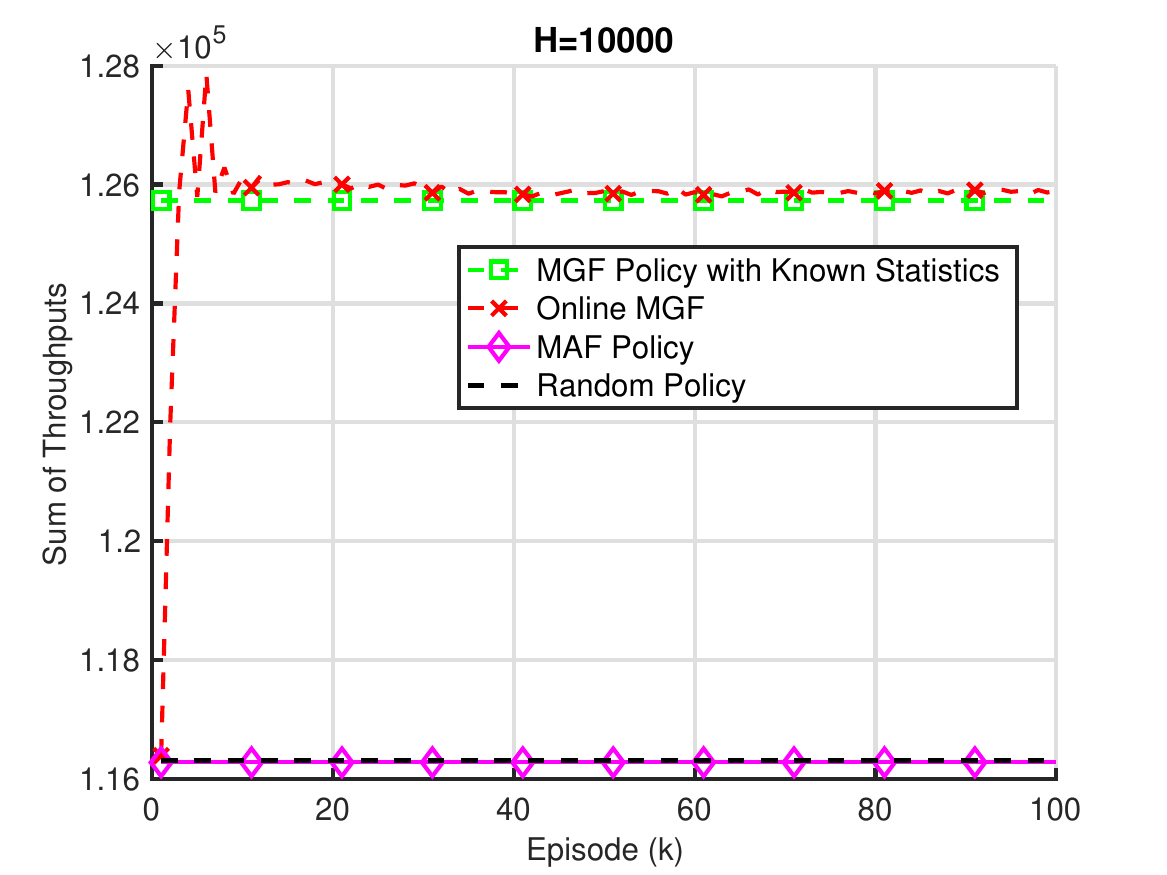}
\end{subfigure}
\vspace{0.05in}
\caption{\small Sum of Throughputs vs the number of episodes ($k$) with different time-horizon length in each episode ($H$). The other parameters are $N=10$, $M=4$, and $P_1$ and $P_2$ are same as the simulation setup.\label{fig:simulationH_three}}
\end{figure*}

\begin{figure*}[t]
  \centering 
  \begin{subfigure}[t]{0.36\textwidth}
\includegraphics[width=\textwidth]{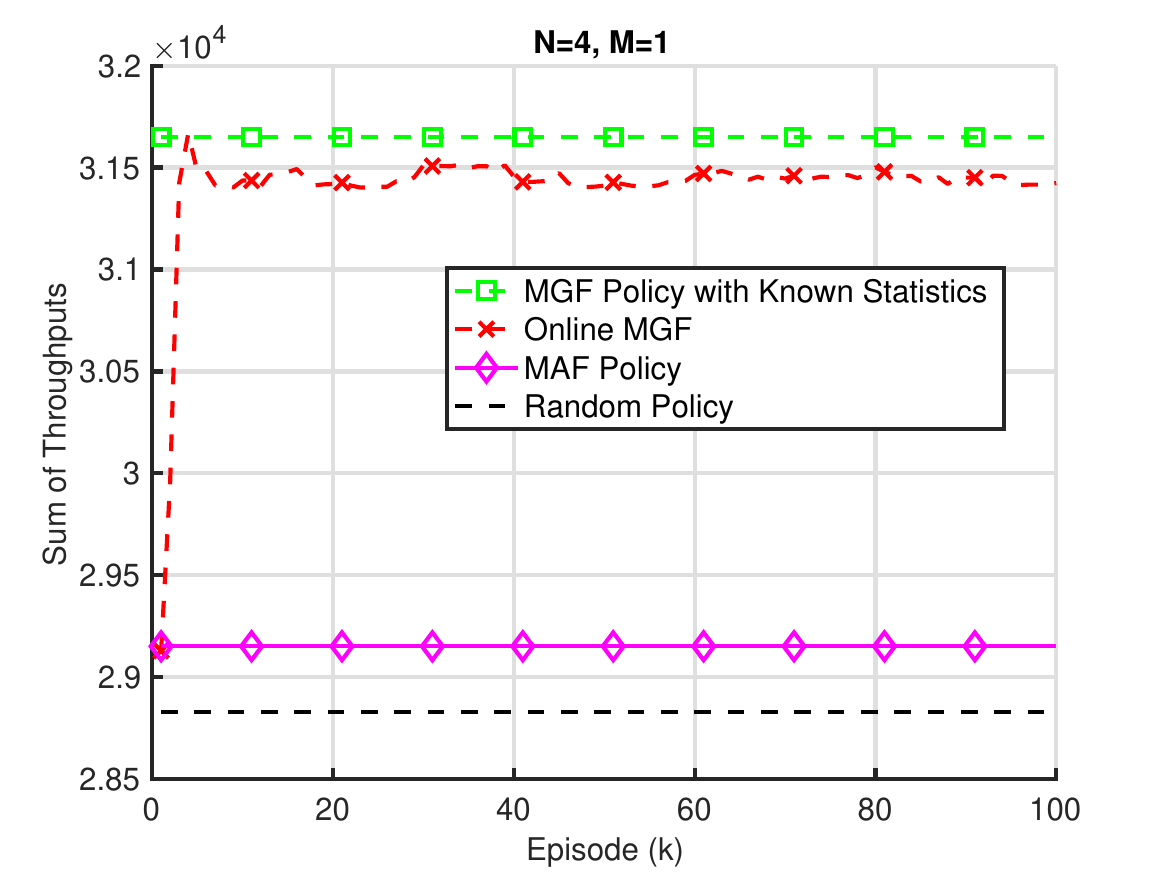}
\end{subfigure}
%
\begin{subfigure}[t]{0.36\textwidth}
\includegraphics[width=\textwidth]{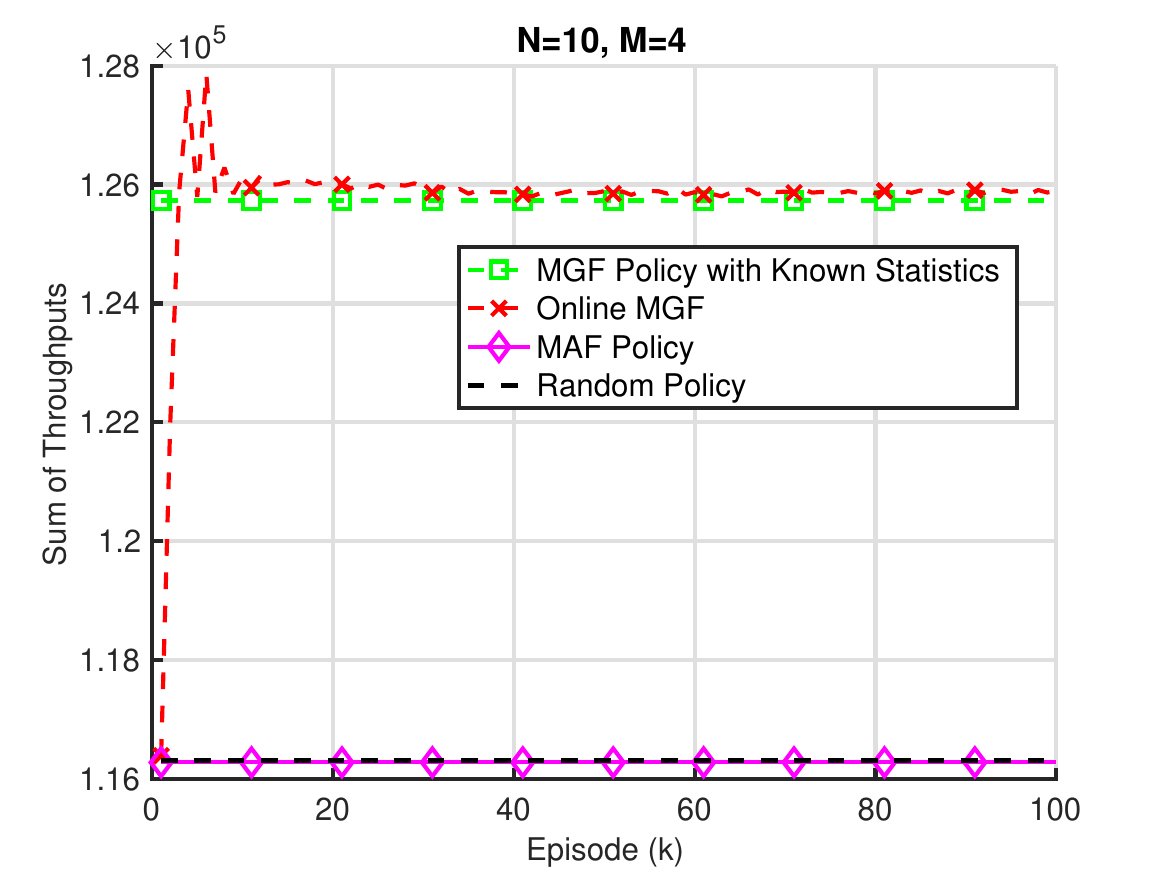}
\end{subfigure}
%
\vspace{0.05in}
\caption{\small Sum of Throughputs vs the number of episodes ($k$) with different configurations of $N$ and $M$. The other parameters are $H=10000$ and $P_1$ and $P_2$ are same as the simulation setup.  \label{fig:simulationN_three}}
\end{figure*}


\subsection{Three-state Channel Model}

In this simulation, we evaluate the performance of the Online-MGF policy for a more complex scenario where the CSI $C_n(t)$ is modeled as a three-state time-homogeneous Markov chain. In this model, the channel states are represented by their Signal-to-Noise Ratio (SNR) values and $C_n (t) \in \{0.5, 6.31, 39.81\}$.
Similar to the ON/OFF channel setup, we consider two classes of users, namely class 1 and class 2. Users in the same class have the same channel statistics, and we consider that half of the users belong to class 1 and the other half belong to class 2. Instead of standard ON/OFF probabilities, the state transitions for users in class 1 and class 2 are governed by distinct $3 \times 3$ transition probability matrices, denoted by $P_1$ and $P_2$, respectively, which are given by
$$P_1 = \begin{bmatrix} 0.80 & 0.20 & 0 \\ 0.10 & 0.80 & 0.10 \\ 0 & 0.20 & 0.80 \end{bmatrix}, \quad P_2 = \begin{bmatrix} 0.50 & 0.40 & 0.10 \\ 0.25 & 0.50 & 0.25 \\ 0.10 & 0.40 & 0.50 \end{bmatrix}.$$
To evaluate scheduling performance, the  instantaneous throughput is calculated by $h(C_n(t), a_n(t)) = a_n(t) \log_2(1 +(C_n(t))$. Because closed-form expression of the Whittle index is not available for three state channel model, Whittle policy with known statistics is not plotted for the figures in this setup.

Figure \ref{fig:simulationH_three} illustrates the sum of throughput vs the number of episodes ($k$) with different time horizon length in each episode $H=1000, 10000$. In this figure, we set $N=10$, and $M=4$. According to Figure \ref{fig:simulationH_three}, when $H$ is small, MGF policy with known channel statistics performs better than the Online-MGF policy. The figure also depicts that the proposed Online-MGF policy converges faster to the MGF policy with known channel statistics as the number of time horizon $H$ in each episode increases because for higher $H$, Online-MGF observes more state transitions in each episode, which yields accurate learning of channel statistics. Moreover, Online-MGF performs better compared to the other baselines, such as MAF and random policies. 

Figure \ref{fig:simulationN_three} illustrates the sum of throughput vs number of episodes ($k$) with different configurations of $N$ and $M$. In this figure, we set $H=10000$ and the other parameters are the same as Figure \ref{fig:simulationH_three}. 
This figure shows that for all configurations, the Online-MGF policy converges faster to the MGF with known channel statistics as the number of users $N$ and channels $M$ increases. As expected, the proposed Online-MGF policy performs better compared to MAF and random policies. 


\section{conclusion}

We study an online wireless scheduling problem with imperfect CSI and unknown channel statistics. The formulated problem is an RMAB. By leveraging the sufficient statistic of the history, we are able to reduce the state space of the underlying RMAB for online learning. To that end, we develop an Online-MGF policy to solve the RMAB that does not need to satisfy any indexability condition and achieves sub-linear regret with the number of episodes.  The proposed policy exhibits less  complexity compared to existing online learning policies. Through simulation, we validate the effectiveness of the proposed Online-MGF policy.
\bibliographystyle{IEEEtran}
\bibliography{ref}
\clearpage
\newpage
\appendices





\section{Proof of Theorem \ref{theorem_simplification}} \label{proof_theorem_1}

We prove Theorem \ref{theorem_simplification} in three steps: (i) First, we show that $C_n (t)$ is conditionally independent of the $(C_n (\tau - \Delta_n (\tau)), \Delta_n (\tau))_{\tau =1} ^ {t-1}$  given the last observed CSI $C_n (t-\Delta_n (t))$ and its AoCSI $\Delta_n (t)$; (ii) Next, we show that $C_n (t)$ is conditionally independent of the scheduling decisions $(a_n (\tau))_{\tau =1} ^ {t-1}$ and the delivery indicators $(\beta_n (\tau))_{\tau =1} ^ {t-1}$ given the AoCSI history $(\Delta_n (\tau))_{\tau =1} ^ {t}$; (iii) Finally, we show that the value function of the finite horizon MDP given the history $\mathcal{H}_n (t)$ is equal to the value function of the finite horizon MDP given $(C_n (t-\Delta_n (t)), \Delta_n (t))$. 

Due to the Markov property and time homogeneity of the CSI $C_n (t)$, we have the following Markov chain given $\Delta_n (t)$: $C_n (t) \leftrightarrow C_n (t-\Delta_n (t)) \leftrightarrow C_n ((t-1)-(\Delta_n (t-1))) \leftrightarrow \ldots$. Therefore, $C_n (t)$ is conditionally independent of $(C_n (\tau - \Delta_n (\tau)), \Delta_n (\tau))_{\tau=1}^{t-1}$ given $(C_n (t-\Delta_n (t)), \Delta_n (t))$. 

Moreover, $C_n (t)$ depends on $(C_n (t-\Delta_n (t)), \Delta_n (t))$ and $(\beta_n (\tau))_{\tau=1}^{t-1}$. To compute $\Delta_n (\tau)$, we require scheduling decisions $(a_n (\tau))_{\tau =0} ^ {t-1}$ and $(\Delta_n (\tau))_{\tau=1}^{t-1}$. Hence, $C_n (t)$ depends on $(a_n (\tau))_{\tau =0} ^ {t-1}$ and $(\Delta_n (\tau))_{\tau=1}^{t-1}$ through $\Delta_n (t)$. However, if the AoCSI history $(\Delta_n (\tau))_{\tau =1} ^ {t}$ is given, then the CSI $C_n (t)$ does not depend on the scheduling decisions $(a_n (\tau))_{\tau =1} ^ {t-1}$ and the delivery indicators $(\beta_n (\tau))_{\tau =1} ^ {t-1}$. Hence, $C_n (t)$ is conditionally independent of the scheduling decisions $(a_n (\tau))_{\tau =0} ^ {t-1}$ and the delivery indicators $(\beta_n (\tau))_{\tau =0} ^ {t-1}$ given the AoCSI history $(\Delta_n (\tau))_{\tau =1} ^ {t}$.

Next, from \eqref{problem}, we can write

\begin{align} \label{exp}
    &\mathbb E\bigg[\sum_{n=1}^{N} h (C_n (t) a_n (t))\bigg|(\mathcal{H}_n (t))_{n=1}^{N}\bigg]\nonumber\\
    =& \sum_{n=1}^{N} \mathbb E\bigg[ h (C_n (t) a_n (t))\bigg|\mathcal{H}_n (t)\bigg],
\end{align}
where \eqref{exp} holds because given $(a_n (\tau))_{\tau=1}^{t-1}$, $(\beta_n(t))_{\tau=1}^{t-1}$ are independent across both users and time-slots and the CSI \{$C_n (t), t = 0,1,2,\ldots, T$\} and \{$C_m (t), t = 0,1,2,\ldots, T$\} are independent for all $n\neq m$. 

Next, utilizing the fact that $C_n (t)$ is conditionally independent of $\mathcal{H}_n (t)$ given $(C_n (t-\Delta_n (t)), \Delta_n (t))$, we get that
\begin{align} \label{exp_1}
    & \mathbb E\bigg[ h (C_n (t) a_n (t))\bigg|\mathcal H_n (t)\bigg] \nonumber\\
    =& \mathbb E\bigg[ h (C_n (t) a_n (t))\bigg|C_n (t-\Delta_n (t)), \Delta_n (t)\!\bigg]\!.
\end{align}

Therefore, the value function of the finite horizon MDP \eqref{problem}-\eqref{constraint} under any given policy $\pi=(a_n(1), a_n(2), \ldots, a_n (T))_{n=1}^{N}$ can be written as
\begin{align} \label{Jpi}
& \sum_{t=1}^{T} \sum_{n=1}^{N} \mathbb E\bigg[ h(C_n (t) a_n (t))\bigg|(\mathcal H_n (t))_{n=1}^{N}\bigg] \nonumber\\
=& \sum_{t=1}^{T} \sum_{n=1}^{N} \mathbb E \bigg[ h(C_n (t) a_n (t)\bigg|C_n (t-\Delta_n (t)), \Delta_n (t)\bigg],
\end{align}
where \eqref{Jpi} holds from \eqref{exp_1}.
Therefore, from \cite[Chapter 4.3]{bertsekas2011dynamic}, Theorem \ref{theorem_simplification} follows.

\section{Proof of Lemma \ref{lemma1}}
\label{proof_lemma1}

To bound the probability that the actual transition $P$ fall outside the confidence region $\mathcal{B}^k$, we evaluate the distance between the true transition probabilities $P_n(\cdot|c,\delta)$ and the empirical transition probabilities $\hat{P}_n^k(\cdot|c,\delta)$.

The $L_1$-deviation between an actual probability distribution and its empirical estimate over $m$ distinct events from $n$ samples is bounded by \cite[Theorem 2.1]{weissman2003inequalities}
\begin{align} \label{standard}
\text{Pr}(|p-\hat p|_1 \geq \alpha) \leq (2^m -2) e^{-\frac{n \alpha^2}{2}}.
\end{align}
We will utilize this result to compare the actual transition probabilities $P_n (\cdot|c, \delta)$ and the optimistic transition probabilities $P_n ^k (\cdot|c, \delta)$ for all state $(c, \delta)$. 

In our model, when an active action is taken, the AoCSI drops to $\delta' = 1$. Therefore, the next state only varies over the size of the CSI space, i.e., $m = |\mathcal{C}|$ distinct possible next states. For any given user $n$ and state $(c, \delta)$, let $n = \max\{1, B_n^k(c, \delta)\}$ denote the number of samples. Substituting the confidence radius $\rho_n ^k (c, \delta)$ defined in \eqref{radius} for $\alpha$ in \eqref{standard}, we get
\begin{align}
& \text{Pr}\left(|\hat{P}_n^k(\cdot|c,\delta) - P_n(\cdot|c,\delta)|_1 \geq \rho_n^k(c,\delta)\right) \nonumber\\
\leq & 2^{|\mathcal{C}|} e^{\left(-\frac{n}{2} \left[ \frac{2|\mathcal{C}| \log(2\Gamma|\mathcal{C}|N \frac{k^2}{\mu})}{n} \right] \right)} \nonumber \\
=& 2^{|\mathcal{C}|} e^{\left(-|\mathcal{C}| \log\left(\frac{2\Gamma|\mathcal{C}|N k^2}{\mu}\right) \right)} \nonumber \\
=& 2^{|\mathcal{C}|} \left( \frac{\mu}{2\Gamma|\mathcal{C}|N k^2} \right)^{|\mathcal{C}|} 
= \left( \frac{\mu}{\Gamma|\mathcal{C}|N k^2} \right)^{|\mathcal{C}|}. \label{single_state_bound}
\end{align}
Because the state space size $|\mathcal{C}| \geq 2$ and the term inside the parenthesis is strictly less than 1 (for sufficiently small $\mu$), we can write
\begin{equation} \label{simplified_bound}
\text{Pr}\left(|\hat{P}_n^k(\cdot|c,\delta) - P_n(\cdot|c,\delta)|_1 \geq \rho_n^k(c,\delta)\right) \leq \frac{\mu}{\Gamma|\mathcal{C}|N k^2}.\end{equation}
The actual transition matrix $\mathbf{P}$ lies outside the confidence region $\mathcal{B}^k$ if the deviation exceeds the confidence radius for any user $n \in \{1, \dots, N\}$ or any state $(c, \delta)$ within the truncated state space size $\Gamma|\mathcal{C}|$. Applying the Union Bound over all $N$ users and all $\Gamma|\mathcal{C}|$ states yields
\begin{align}
&\text{Pr} (\mathbf{P} \notin \mathcal{B}^k) \nonumber\\
\leq & \sum_{n=1}^N \sum_{(c, \delta) \in \mathcal{C} \times \mathcal{D}} \text{Pr} \left(|\hat{P}_n^k(\cdot|c,\delta) - P_n(\cdot|c,\delta)|_1 \ge \rho_n^k(c,\delta)\right) \nonumber\\ 
\leq & N  \Gamma|\mathcal{C}|  \left( \frac{\mu}{\Gamma|\mathcal{C}|N k^2} \right)  
= \frac{\mu}{k^2}. \label{union_bound}
\end{align}
Therefore, the probability that the true transition matrix $\mathbf{P}$ falls within the confidence ball is
\begin{align}
\text{Pr}(\mathbf{P} \in \mathcal{B}^k) = 1 - \text{Pr}(\mathbf{P} \notin \mathcal{B}^k) \ge 1 - \frac{\mu}{k^2}.
\end{align}
This completes the proof.

\ignore{
In this sequel, we get
\begin{align}
\text{Pr}(|P_n ^k (\cdot|c, \delta)-P_n (\cdot|c, \delta)|_1 \geq \alpha) \leq (2^{\tau|\mathcal{C}|} -2) e^{-\frac{n \alpha^2}{2}}. \label{L1_dev}
\end{align}
Let $\alpha = \sqrt{\frac{2}{n} \log\bigg(2^{\tau |\mathcal{C}|} \tau |\mathcal{C}| N \frac{k^2}{\mu} \bigg)}$. Substituting $\alpha$ in \eqref{L1_dev} yields
\begin{align}
\text{Pr}\bigg(|P_n ^k (\cdot|c, \delta)-P_n (\cdot|c, \delta)|_1 & \geq \sqrt{\frac{2}{n} \log\bigg(2^{\tau |\mathcal{C}|} \tau |\mathcal{C}| N \frac{k^2}{\mu} \bigg)}\bigg) \nonumber\\
& \leq 2^{\tau|\mathcal{C}|} e^{-\log\bigg(2^{\tau |\mathcal{C}|} \tau |\mathcal{C}| N \frac{k^2}{\mu} \bigg)} \nonumber\\
& = \frac{\mu}{\tau |\mathcal{C}| N k^2}
\end{align}
Denote $n = \max\{1, B_n ^k (c, \delta)\}$ for all $(c, \delta)$. Next, by taking summation over all states $(c, \delta)$ and all users $n \in \{1,2, \ldots, N\}$ implies
\begin{align}
\text{Pr} (\mathbf{P} \notin \mathcal{B} ^{k}) \leq \frac{\mu}{k^2}.
\end{align}
Therefore,
\begin{align}
\text{Pr} (\mathbf{P} \in \mathcal{B} ^{k}) \geq 1-\frac{\mu}{k^2}.
\end{align}
This completes the proof.}

\section{Proof of Lemma \ref{lemma2}}
\label{proof_lemma2}

If the confidence bound holds, i.e., $\mathbf{P} \in \mathcal{B}^k$, we can write
\begin{align}
& \text{Reg}_{\text{comp}_1} (k) \nonumber\\
=& L(\pi^*, \mathbf{P}, \lambda^*) - L(\pi^k, \mathbf{P}, \lambda^k) \nonumber \\
= & \sum_{n=1}^N \left( V_{n,1}^{P_n}(c_n(1), \delta_n(1); \lambda^*) - V_{n,1}^{P_n}(c_n(1), \delta_n(1); \lambda^k; \pi_n ^k) \right), \nonumber\\
\leq & \sum_{n=1}^N \left( V_{n,1}^{P_n}(c_n(1), \delta_n(1); \lambda^k) - V_{n,1}^{P_n}(c_n(1), \delta_n(1); \lambda^k; \pi_n ^k) \right), \label{regret_value_bound}
\end{align}
where \eqref{regret_value_bound} holds because the Lagrange multiplier $\lambda^*$ minimizes the Lagrangian $L(\pi^*, \mathbf{P}, \lambda)$. Note that $V_{n,1}^{P_n}(\cdot; \pi_n ^k)$ represents the value achieved under the actual transition probabilities by executing policy $\pi^k$.

Next, in order to prove \eqref{lemma2_3}, we first show that the optimal value function under the true transition $P_n$ is upper-bounded by the optimistic value function under $P_n^k$, which is given by
\begin{align} \label{value_upper_bound}
V_{n,t}^{P_n}(c, \delta; \lambda^k) \leq V_{n,t}^{P_n^k}(c, \delta; \lambda^k).
\end{align}
We prove \eqref{value_upper_bound} by using backward induction. For $t = H+1$, it is trivially true since $V_{n,H+1}^{P_n}(c, \delta) = V_{n,H+1}^{P_n^k}(c, \delta) = 0$ for all $(c, \delta)$.

Let \eqref{value_upper_bound} holds for $t = h+1$. Then, we have
\begin{align} \label{inductive_hypothesis}
V_{n,h+1}^{P_n}(c, \delta; \lambda^k) \leq V_{n,h+1}^{P_n^k}(c, \delta; \lambda^k).
\end{align}
We have to show that if \eqref{inductive_hypothesis} holds, then the following is true for $t = h$,
\begin{align}
& V_{n,h}^{P_n}(c, \delta; \lambda^k) - V_{n,h}^{P_n^k}(c, \delta; \lambda^k) \nonumber\\
=& \max_{a \in \{0,1\}} Q_{n,h}^{P_n}(c, \delta, a; \lambda^k) - \max_{a \in \{0,1\}} Q_{n,h}^{P_n^k}(c, \delta, a; \lambda^k) \nonumber \\
\leq & \max_{a \in \{0,1\}} \bigg( Q_{n,h}^{P_n}(c, \delta, a; \lambda^k) - Q_{n,h}^{P_n^k}(c, \delta, a; \lambda^k) \bigg). \label{max_difference}
\end{align}
Next, we analyze \eqref{max_difference} for both actions. If the passive action ($a=0$) is taken, the last observed CSI will not change, and the state transition is deterministic. Using  \eqref{inductive_hypothesis}, we get
\begin{align}
& Q_{n,h}^{P_n}(c, \delta, 0; \lambda^k) - Q_{n,h}^{P_n^k}(c, \delta, 0; \lambda^k) \nonumber\\
=& V_{n,h+1}^{P_n}(c, \delta+1; \lambda^k) - V_{n,h+1}^{P_n^k}(c, \delta+1; \lambda^k) \leq 0.
\end{align}
If the active action ($a=1$) is taken, we have
\begin{align}
& Q_{n,h}^{P_n}(c, \delta, 1; \lambda^k) - Q_{n,h}^{P_n^k}(c, \delta, 1; \lambda^k) \nonumber\\
=& \sum_{c'} P_n(c'|c,\delta)\big(h(c',1) + V_{n,h+1}^{P_n}(c',1;\lambda^k)\big) - \nonumber\\
& \max_{\tilde{P}_n ^k \in \mathcal{B}^k} \sum{c'} \tilde{P}n^k(c'|c,\delta)\big(h(c',1) + V_{n,h+1} ^{P_n^k}(c',1;\lambda^k)\big) 
\nonumber\\
\leq & \sum_{c'} P_n(c'|c,\delta)\big(h(c',1) + V_{n,h+1}^{P_n^k}(c',1;\lambda^k)\big) - \nonumber\\ & \max_{\tilde{P}_n ^k \in \mathcal{B}^k} \sum{c'} \tilde{P}n^k(c'|c,\delta)\big(h(c',1) + V_{n,h+1}^{P_n^k}(c',1;\lambda^k)\big) \label{active_diff_step1} \\
\leq & 0. \label{active_diff_step2}
\end{align}
Here, \eqref{active_diff_step1} holds from $V_{n,h+1}^{P_n} \leq V_{n,h+1}^{P_n^k}$ and \eqref{active_diff_step2}  holds because the actual transition probability $P_n$ lies within the confidence ball $\mathcal{B}^k$ (i.e., $P_n \in \mathcal{B}^k$). Therefore, the expected value  at $P_n$ can never exceed the maximum possible expected value found by searching over all $\tilde{P}_n^k \in \mathcal{B}^k$. Because this implies the difference between the true expectation and the optimistic expectation is less than or equal to zero for all actions, we conclude $V_{n,h}^{P_n}(c, \delta; \lambda^k) \le V_{n,h}^{P_n^k}(c, \delta; \lambda^k)$ for all $t$. Substituting this into \eqref{regret_value_bound} yields
\begin{align}
& \text{Reg}_{\text{comp}_1}(k) \nonumber\\ \leq & \sum_{n=1}^N \left( V_{n,1}^{P_n^k}(c_n(1), \delta_n(1); \lambda^k) - V_{n,1}^{P_n}(c_n(1), \delta_n(1); \lambda^k; \pi_n ^k) \right).
\end{align}
This completes the proof.

\section{Proof of Theorem \ref{theorem_regret}}
\label{proof_theorem_regret}

We bound the cumulative regret of $\text{comp}_1$ over $K$ episodes by analyzing two cases: when the true transition probability $\mathbf{P}$ lies within the confidence ball $\mathcal{B}^k$, and when it does not.

When the confidence bound holds (i.e., $\mathbf{P} \in \mathcal{B}^k$), from Lemma \ref{lemma2}, the regret for episode $k$ is bounded by the difference in value functions evaluated under the policy $\pi^k$, which is given by
\begin{align}
& \text{Reg}_{\text{comp}_1} (k) \nonumber\\
\leq & \sum_{n=1}^N \left( V_{n,1}^{P_n^k}(c_n(1), \delta_n(1); \lambda^k) - V_{n,1}^{P_n}(c_n(1), \delta_n(1); \lambda^k; \pi_n ^k) \right).
\end{align}
Since we need to estimate the transitions only for the active actions, by using Lemma \ref{lemma2} and applying Lemma 10 of \cite{ju2023achieving}, we get
\begin{align}
& V_{n,1}^{P_n^k}(c_n(1), \delta_n(1); \lambda^k) - V_{n,1}^{P_n}(c_n(1), \delta_n(1); \lambda^k; \pi_n ^k) \nonumber\\
=& \mathbb{E}_{P_n, \pi_n ^k} \!\! \bigg[\!\sum_{t=1}^{H} \!\!\mathbf{1}(a_n(t)\!=\!1)\!\! \sum_{c', \delta'} \bigg( P_n^k(c', \delta'|c_n(t), \delta_n(t)) \nonumber\\
& ~~~~~~~~~~~~~~~~~~~~~~~~~~~~~ -\!\! P_n(c', \delta'|c_n(t), \delta_n(t)) \bigg) \!V_{n,t+1}^{P_n^k}\!(s'\!) \!\bigg]. \label{value_difference_1}
\end{align}
Applying H\"{o}lder's inequality \cite{agarwal2019reinforcement} and utilizing the upper bound of the value function $V_{\max}$, we obtain
\begin{align}
& V_{n,1}^{P_n^k}(c_n(1), \delta_n(1); \lambda^k) - V_{n,1}^{P_n}(c_n(1), \delta_n(1); \lambda^k; \pi_n ^k) \nonumber\\
\leq & V_{\max} \mathbb{E}_{P_n, \pi_n ^k} \bigg[ \sum_{t=1}^{H} \mathbf{1}(a_n(t)=1) \bigg| P_n^k(\cdot|c_n(t), \delta_n(t)) \nonumber\\
&~~~~~~~~~~~~~~~~~~~~~~~~~~~~~~~~~~~~~ - P_n(\cdot|c_n(t), \delta_n(t)) \bigg|_1 \bigg].
\end{align}
Because $P_n \in \mathcal{B}^k$ and $P_n^k \in \mathcal{B}^k$, both transition matrices lie within the confidence radius $\rho_n^k$ with center at the empirical estimate $\hat{P}_n^k$. Using the triangle inequality, we can write
\begin{align}
& \big|P_n^k(\cdot|c_n(t), \delta_n(t)) - P_n(\cdot|c_n(t), \delta_n(t)) \big|_1 \nonumber\\
\leq & \big|P_n^k(\cdot|c_n(t), \delta_n(t)) - \hat{P}_n^k(\cdot|c_n(t), \delta_n(t)) \big|_1 + \nonumber\\
& \big| \hat{P}_n^k(\cdot|c_n(t), \delta_n(t)) - P_n(\cdot|c_n(t), \delta_n(t)) \big|_1 \nonumber\\ 
\leq & 2\rho_n^k(c_n (t), \delta_n (t)). \label{value_difference_3}
\end{align}
Therefore, when the confidence bound holds, the regret for episode $k$ is bounded by
\begin{align} \label{per_episode_rho}
& \text{Reg}_{\text{comp}_1}(k) \nonumber\\
\leq & 2 V_{\max} \sum_{n=1}^{N} \mathbb{E}_{P_n, \pi_n ^k} \left[ \sum_{t=1}^{H} \mathbf{1}(a_n(t)=1) \rho_n^k(c_n (t), \delta_n (t)) \right].
\end{align}
Let $\gamma_n ^k (c, \delta)$ is a random variable, which represents the number of visits to state $(c, \delta) \in \mathcal{C} \times \mathcal{D}$ in episode $k$ for arm $n$ and the active action is taken by arm $n$ during episode $k$. Substituting $\rho_n^k(c, \delta)$ form \eqref{radius} to \eqref{per_episode_rho}, we get the total regret over $K$ episodes as follows
\begin{align} \label{ref_proof_1}
& \sum_{k=1}^K \text{Reg}_{\text{comp}_1}(k) \nonumber\\
\leq & 2 V_{\max} \!\!\sum_{n=1}^N \mathbb{E} \bigg[\sum_{k=1}^K \!\!\sum_{(c', \delta') \in \mathcal{C} \times \mathcal{D}} \!\!\!\!\!\gamma_n^k(c, \delta) \sqrt{\frac{2|\mathcal{C}|\log(2\Gamma|\mathcal{C}|N \frac{k^2}{\mu})}{\max(1, B_n^k(c, \delta))}} \bigg].
\end{align}
Next, by applying \cite[Lemma E.3]{wang2023optimistic} yields
\begin{align}
& \sum_{k=1}^{K} \sum_{(c,\delta) \in \mathcal{C} \times \mathcal{D}} \frac{\gamma_n ^{k} (c, \delta)}{\sqrt{\max({1, B_n ^{k} (c, \delta))}}} \nonumber\\
\leq & \bigg(\sqrt{H+1} + 1\bigg) \sum_{(c,\delta) \in \mathcal{C} \times \mathcal{D}} \sqrt{B_n ^{k} (c, \delta)} \nonumber\\
\leq & \bigg(\sqrt{H+1} + 1\bigg) \sqrt{\Gamma |\mathcal{C}| K H},
\end{align}
where the last inequality holds because the total number of active actions taken by arm $n$ over $K$ episodes is at most $K H$.
Substituting this into \eqref{ref_proof_1}, we get the cumulative regret
\begin{align}
& \sum_{k=1}^K \text{Reg}_{\text{comp}_1}(k) \nonumber\\
&\leq 2 V_{\max} \sqrt{2|\mathcal{C}|\log\left(2\Gamma|\mathcal{C}|N \frac{K^2}{\mu}\right)} \nonumber\\
&~~~\sum_{n=1}^{N} \left((\sqrt{H+1} + 1) \sqrt{\Gamma |\mathcal{C}| K H} \right) \nonumber\\
&= \mathcal{O}\left( V_{\max} \Gamma |\mathcal{C}| N H \sqrt{K \log K} \right). \label{final_good_event}
\end{align}

When the confidence bound fails (i.e, $\mathbf{P} \notin \mathcal{B}^k$), from Lemma \ref{lemma1}, we know $\text{Pr}(\mathbf{P} \notin \mathcal{B}^k) \leq \mu/k^2$. The maximum possible regret in any episode is bounded by $2 V_{\max} N H$. Summing over $K$ episodes, the cumulative regret when the confidence bound does not hold is given by
\begin{align}
\sum_{k=1}^K 2 V_{\max} N H \text{Pr} (\mathbf{P} \notin \mathcal{B}^k) & \leq 2 V_{\max} N H \sum_{k=1}^K \frac{\mu}{k^2} \nonumber\\
& \leq 2 V_{\max} N H \mu \frac{\pi^2}{6} \nonumber\\
& = \mathcal{O}(\mu).
\end{align}
Combining both cases, the total cumulative regret for $\text{comp}_1$ holds with probability $1-\mu$ and is given by
\begin{align} \label{final_reg_1}
\text{Reg}_{\text{comp}_1}(K) \leq& \mathcal{O}\left( V_{\max} \Gamma |\mathcal{C}| N H \sqrt{K \log K} \right) + \mathcal{O}(\mu) \nonumber\\
=& \mathcal{O}\left( V_{\max} \Gamma |\mathcal{C}| N H \sqrt{K \log K} \right).
\end{align}
This completes the proof.

\section{Proof of Lemma \ref{lemma_g}} \label{proof_lemma_g}

Following the Lagrangian bounding framework established in \cite{brown2020index}, the sub-optimality of the MGF policy is bounded by the expected sum of single-step action-value difference relative to the optimal actions $a_n^*(t)$ taken to solve the relaxed problem  evaluated at the optimal dual cost $\lambda^*$. 
For episode $k$, the cost due to this difference is bounded by
\begin{align} \label{brown_gap}
g^k (N) \leq \mathbb{E}_{\pi^{\text{MGF}}} \bigg[\sum_{t=1}^H \sum_{n=1}^N \bigg(& Q_{n,t}(c_n(t), \delta_n (t), a_n^*(t)) \nonumber\\
& - Q_{n,t} (c_n(t), \delta_n (t), a_n^{\text{MGF}}(t)) \bigg) \bigg],
\end{align}
where $a_n^*(t)$ denote the action chosen to solve the relaxed problem \eqref{relax_problem}-\eqref{constraint_relax}, and $a_n^{\text{MGF}}(t)$ denote the action chosen according to the MGF policy.

Using standard fluid approximations, \cite{brown2020index} bound this difference. By leveraging the gain index in \eqref{gain_idx}, we can directly compute this difference. From \eqref{gain_idx}, the optimal policy for the relaxed problem selects $a_n ^* = 1$ if $\alpha_n(c, \delta) > \lambda^*$, and the passive action otherwise. The MGF policy selects $M$ arms with the highest gain $\alpha_n(c, \delta)$. Because the channel state is finite, the maximum single-step cost is bounded by $\Delta_{\max}$, which depends entirely on the maximum CSI $C_{\max}$ and the maximum AoCSI bound $\tau$
\begin{align} \label{our_delta_bound}
\max_{n, c, \delta} |\alpha_n(c, \delta) - \lambda^*| \leq \Delta_{\max}.
\end{align}
Substituting \eqref{our_delta_bound} into \eqref{brown_gap}, we get
\begin{align}
g^k(N) \leq\Delta_{\max} \sum_{t=1}^H \mathbb{E}_{\pi^{\text{MGF}}} \left[\sum_{n=1}^N \mathbf{1}(a_n^*(t) \neq a_n^{\text{MGF}}(t)) \right].
\end{align}

Based on the fluid limit convergence for RMABs, as the network size $N \to \infty$, the empirical state distribution converges to a deterministic fluid limit \cite{weber1990index}. The deviation of the strictly feasible constraint ($M$ activations) from the relaxed threshold constraint ($\lambda^*$) behaves according to the Central Limit Theorem \cite{weber1990index, brown2020index}. Specifically, the expected total number of mismatched arms scales with the standard deviation of the empirical measure, bounded by $C_{v} \sqrt{N}$ (where $C_{v}$ is the variance constant) \cite{brown2020index}.

The expected mismatch for each user in each episode is given by
\begin{align}
\frac{1}{N} g^k (N) \leq \frac{1}{N} \Delta_{\max} H C_{v} \sqrt{N} = \frac{H C_{v} \Delta_{\max}}{\sqrt{N}}.
\end{align}
Next, by taking summation over the finite horizon of $K$ episodes, we get
\begin{align}
g(N) \leq \sum_{k=1}^K \frac{H C_{v} \Delta_{\max}}{\sqrt{N}} = \mathcal{O}\left( \frac{K}{\sqrt{N}} \right).
\end{align}
This completes the proof.

\section{Proof of Theorem \ref{theorem_regret_2}} \label{proof_theorem_regret_2}

Utilizing $g (N)$ in \eqref{g}, we  can decompose $\text{comp}_2$ as follows
\begin{align}
\text{comp}_2 =& \sum_{n=1}^{N} L_n (\pi_n ^{k}, P_n, \lambda^*) - J^k (\pi^k, \mathbf{P}) \nonumber\\
=& g(N) + \sum_{n=1}^{N} L_n(\pi_n ^{k}, P_n, \lambda^*) - \sum_{n=1}^{N} L_n(\pi_n ^{k}, P_n ^k, \lambda^*) \nonumber\\
& + J^k (\pi^{k}, \mathbf{P}^k) - J^k (\pi^k, \mathbf{P}). \label{reg_new_term2}
\end{align}
If the confidence bound holds, then
\begin{align} \label{regret_cumulative}
\sum_{n=1}^{N} L_n(\pi_n ^{k}, P_n, \lambda^*) -  L_n(\pi_n ^{k}, P_n ^k, \lambda^*) \leq 0.
\end{align}
This is because $\pi_n ^k$ is the optimal solution to the Lagrangian problem and $\mathbf{P}^k$ achieves the highest Lagrangian objective value.

Furthermore, following the value difference bound established in \eqref{value_difference_1}-\eqref{value_difference_3}, the difference in the expected cumulative throughput under the same policy $\pi^k$ but different transition probabilities is bounded by 
\begin{align}
& J^k(\pi^k, \mathbf{P}^k) - J^k(\pi^k, \mathbf{P}) \nonumber\\
\leq & \sum_{n=1}^N V_{\max} \mathbb{E}_{P_n, \pi_n ^k} \bigg[ \sum_{t=1}^H \mathbf{1}(a_n(t)=1) \big| P_n^k(\cdot|c_n(t),\delta_n(t))\nonumber\\ 
& ~~~~~~~~~~~~~~~~~~~~~~~~~~~~~~~~~~~~~~- P_n(\cdot|c_n(t), \delta_n(t)) \big|_1 \bigg] \nonumber \\
\leq & 2 V_{\max} \sum_{n=1}^N \mathbb{E}_{P_n, \pi_n ^k} \left[ \sum_{t=1}^H \mathbf{1}(a_n(t)=1) \rho_n^k(c_n(t), \delta_n(t)) \right]. \label{J_diff_bound}
\end{align}

By taking summation of  \eqref{J_diff_bound} over all $K$ episodes and following the similar procedure in Appendix \ref{proof_theorem_regret}, we get
\begin{align} \label{regret_2_result}
& J^k (\pi^{k}, \mathbf{P}^k) - J^k (\pi^k, \mathbf{P}) \nonumber\\
\leq & \mathcal{O} \bigg(2 V_{\max} \Gamma |\mathcal{C}| N H \sqrt{K \log K}\bigg).
\end{align}


When the confidence bound fails (i.e., $\mathbf{P} \notin \mathcal{B}^k$), which occurs with probability at most $\mu/k^2$, the maximum cost per episode is bounded by $2V_{\max}NH$. Taking summation over $K$ episodes yields the bound $\mathcal{O}(\mu)$.

Utilizing \eqref{regret_cumulative} and \eqref{regret_2_result} in \eqref{reg_new_term2} for all $K$ episodes, with probability $1-\mu$, the cumulative regret associated with $\text{comp}_2$ is given by
\begin{align} \label{regret_proof_3}
\text{Reg}_{\text{comp}_2} \leq g(N) K + \mathcal{O} \bigg(2 V_{\max} \Gamma |\mathcal{C}| N H \sqrt{K \log K}\bigg). 
\end{align}
Finally, utilizing Lemma \ref{lemma_g} and taking $N \to \infty$ yields \eqref{theorem_regret_eq1}. 

This completes the proof.

\section{Proof of Theorem \ref{theorem_regret_3}}
\label{proof_theorem3}

Using \eqref{final_reg_1} and \eqref{regret_proof_3}, we get the cumulative regret over all $K$ episodes as follows
\begin{align}
\text{Reg} (K) & \leq \text{Reg}_{\text{comp}_1} + \text{Reg}_{\text{comp}_2} \nonumber\\
& \leq g(N) K + O\bigg(3 V_{\max} \Gamma |\mathcal{C}| N H \sqrt{K \log K}\bigg),    
\end{align}
from which \eqref{theorem_regret_eq1} follows by using Lemma \ref{lemma_g} and taking $N \to \infty$.

This completes the proof.

\section{Proof of Theorem \ref{theorem_5}} \label{indexability_proof}


To prove indexability for ON/OFF channels, we consider the decoupled problem \eqref{decoupled} with a subsidy $\lambda$ for the passive action, i.e., when $a=0$ (equivalently, a cost $\lambda$ for the active action). We assume positively correlated channels, i.e., $P_n(1|1,\delta) > P_n(1|0,\delta)$. From Definition \ref{indexability_def}, an arm is indexable if the set $\Phi(\lambda)$ increases monotonically as the activation cost $\lambda$ increases. We establish this by analyzing the threshold policies for both observed states.

\emph{Case 1:} (The last observed state is ``ON" ($c=1$).) 
When $c=1$, the expected immediate reward for taking the active action is $P_n (1 | 1, \delta) - \lambda$. Because the channels are positively correlated, $P_n (1|1, \delta)$ is decreasing with $\delta$. Consequently, the active action becomes less beneficial as $\delta$ increases. the action-value function $Q_{n, t} (1, \delta, 1)$ for the active action decreases as $\delta$ increases. Therefore, the optimal policy forms a threshold structure: take the active action if $\delta \leq \delta_1^*$ and the passive action if $\delta > \delta_1^*$, where $\delta_1^*$ is the AoCSI threshold. As $\lambda$ increases, the net reward for taking the active action decreases with $\delta$ and the active action becomes less beneficial for all $\delta$. Consequently, the threshold $\delta_1^*$ also decreases or remains the same. A decreasing threshold implies that the set of states where the arm is active decreases as $\lambda$ increases. Therefore, the set of passive states $\Phi_1(\lambda) = \{(1, \delta) \mid \delta > \delta_1^*\}$ in \label{def_5} monotonically increases as $\lambda$ increases.



\emph{Case 2:} (The last observed state is ``OFF" ($c=0$).) 
When $c=0$, the expected immediate reward for taking the active action is $P_n(1|0,\delta) - \lambda$. The probability $P_n(1|0,\delta)$ monotonically increases with $\delta$ and growing towards the steady state probability distribution. This implies that the immediate reward for activation also increases as $\delta$ grows. Because the active action becomes more beneficial over time, the optimal policy forms a reversed-threshold structure: take the passive action if $\delta \leq \delta_0^*$ and the active action if $\delta > \delta_0^*$. As the penalty $\lambda$ increases, the immediate net reward of taking the active action decreases for all $\delta$. This delays the point at which activation becomes optimal and the threshold $\delta_0^*$ exhibit larger values. A non-decreasing threshold $\delta_0^*$ implies that the set of passive states $\Phi_0(\lambda) = \{(0, \delta) \mid \delta \le \delta_0^*\}$ monotonically expands as $\lambda$ increases. 

Since the passive sets for both $c=1$ and $c=0$ grow monotonically as the activation cost $\lambda$ increases, all arms are indexable by Definition \ref{indexability_def}. This completes the proof.

\section{Proof of Theorem \ref{theorem_6}} \label{proof_Whittle}

We will derive the closed-form Whittle index by finding the cost $\lambda$ for which the active and passive actions are equally beneficial at a given state $(c, \delta)$, i.e., $Q(c, \delta, 1;\lambda) = Q(c, \delta, 0;\lambda)$. 

\emph{Case 1:} (Whittle index for $c=1$.) To compute the Whittle index at state $(1, \delta)$, we will utilize the action-value functions for both active and passive actions, i.e., $Q (1, \delta, 1;\lambda) = Q (1, \delta, 0;\lambda)$.

If passive action is taken at $(1, \delta)$, then, AoCSI increases to $\delta+1$. Since $P_n (1|1, \delta)$ is decreasing with $\delta$, the optimal policy prefers to stay passive at $\delta+1$. Thus, the passive action leads to a trajectory of staying passive forever. Hence, the action-value function for the passive action relative to $\lambda$ will be 0, i.e., $Q (1, \delta, 0;\lambda) =0$. If the active action is taken, the AoCSI becomes 1. The next state will be $(1,1)$ with probability $P_n (1|1, \delta)$ and $(0,1)$ with probability $P_n (0|1, \delta)$. The action-value function for the active action is then given by
\begin{align} \label{whittle_active}
Q(1, \delta, 1;\lambda) = & P_n(1|1, \delta) - \lambda + P_n(1|1, \delta) V(1, 1;\lambda) + \nonumber\\
& (1 - P_n(1|1, \delta))V(0, 1;\lambda).
\end{align}
Because the channels are positively correlated, if the last observed CSI is ``OFF" (i.e, $c=0$), it is highly likely to remain at the ``OFF" state. This makes the probability $P_n (0|1, \delta)$ very small and yields
$V (0, 1;\lambda) =0$. Conversely, state $(1, 1)$ represents the highest possible probability $P_n (1|1, 1)$. Since $\delta \geq 1$, the state $(1, 1)$ strictly belongs to the active region. Then the value function $V(1,1;\lambda)$ is given by
\begin{align}
V (1, 1; \lambda) = & P_n (1|1, 1) -\lambda + P_n (1|1, 1) V(1, 1; \lambda) \nonumber\\
& (1-P_n (1|1,1)) V(0,1;\lambda).
\end{align}
which yields
\begin{align}
V (1, 1; \lambda) = \frac{P_n (1|1, 1) -\lambda}{1 - P_n (1|1, 1)},
\end{align}
which holds because $V(0,1) =0$.
Substituting the values of $V(1, 0; \lambda)$ and $V(1, 1; \lambda)$ in \eqref{whittle_active} and $Q(1, \delta, 1; \lambda) = Q(1, \delta, 0; \lambda) =0$, we get
\begin{align}
P_n (1|1, \delta) - \lambda + P_n (1|1, \delta) \bigg(\frac{P_n (1|1, 1) - \lambda}{1 - P_n (1|1, 1)} \bigg) = 0.
\end{align}
Solving for $\lambda$ yields the Whittle index $W_n (c, \delta)$ as follows
\begin{align}
W_n (c, \delta) = \frac{P_n (1|1, \delta)}{1 - (P_n (1|1, 1) - P_n (1|1, \delta) )}.
\end{align}

\emph{Case 2:} (Whittle index for $c=0$.) Due to the positive correlation, if the last observed state is ``OFF," the probability of the channel being ``ON" starts at its lowest possible value and slowly rises. Throughout this trajectory, the expected reward is consistently lower than the expected reward of an arm starting at $c=1$. 
Hence, it is never optimal to schedule and the Whittle index is 0.

Combining cases 1 and 2, we get \eqref{whittle_index}. This completes the proof.

\end{document}